\documentclass[a4paper,11pt]{article}

\usepackage{fullpage}
\usepackage{times}
\usepackage{soul}
\usepackage{url}
\usepackage[hidelinks]{hyperref}
\usepackage[utf8]{inputenc}
\usepackage[small]{caption}
\usepackage{graphicx}
\usepackage[table]{xcolor}
\usepackage{amsmath}
\usepackage{booktabs}
\usepackage{multirow, hhline}
\usepackage[ruled,vlined,linesnumbered]{algorithm2e}
\usepackage{algorithmic}
\usepackage{enumitem}
\usepackage{multirow}
\usepackage{makecell}
\usepackage{tcolorbox}
\usepackage{amsthm}
\usepackage{amssymb}

\usepackage{bbm}
\usepackage{authblk}

\newtheorem{theorem}{Theorem}[section]

\newtheorem{proposition}[theorem]{Proposition}
\newtheorem{lemma}[theorem]{Lemma}

\theoremstyle{definition}
\newtheorem{definition}[theorem]{Definition}
\newtheorem{example}[theorem]{Example}

\usepackage{natbib}

\allowdisplaybreaks

\newcommand{\numbox}[1]{{\setlength\fboxrule{-1.5pt}\fbox{\tcbox[colframe=black,colback=gray!15,shrink tight,boxrule=0.5pt,extrude by=0.5mm]{\scriptsize #1}}}}
\newcommand{\numboxblk}[1]{{\setlength\fboxrule{-1.5pt}\fbox{\tcbox[colframe=black,colback=gray!50,coltext=black,shrink tight,boxrule=0.5pt,extrude by=0.5mm]{\scriptsize #1}}}}
\newcommand{\threecand}{$\numbox{3}$}
\newcommand{\bdzero}{$\numbox{0}$}
\newcommand{\uncon}{$\numbox{$\infty$}$}
\newcommand{\unary}{$\numboxblk{\text{U}}$}

\newcount\Comments  
\newcommand{\kibitz}[2]{\ifnum\Comments=1{\color{#1}{#2}}\fi}

\newcommand{\vv}{\mathbf{v}}

\newcommand{\tildevv}{\tilde{\mathbf{v}}}
\newcommand{\uu}{\mathbf{u}}
\newcommand{\tildeuu}{\tilde{\mathbf{u}}}

\newcommand{\SW}{\text{SW}}
\newcommand{\PV}{\mathrm{PV}}
\newcommand{\PD}{\mathrm{PD}}
\renewcommand{\deg}{\mathrm{deg}}
\newcommand{\calR}{\mathcal{R}}
\newcommand{\calA}{\mathcal{A}}
\newcommand{\calD}{\mathcal{D}}
\newcommand{\calS}{\mathcal{S}}
\newcommand{\calQ}{\mathcal{Q}}

\newcommand{\true}{\text{\em true}}
\newcommand{\false}{\text{\em false}}

\renewcommand{\arraystretch}{1.2}

\title{\bf Protecting Elections by Recounting Ballots\thanks{Edith Elkind and Alexandros A. Voudouris are supported by the European Research Council (ERC) under grant number 639945 (ACCORD). Jiarui Gan is supported by the EPSRC International Doctoral Scholars Grant EP/N509711/1. Svetlana Obraztsova is supported by the MOE AcRF-T1-RG23/18 grant. Zinovi Rabinovich is supported by the NTU SUG M4081985 grant.
}}

\author[1]{Edith Elkind}
\author[1]{Jiarui Gan}
\author[2]{Svetlana Obraztsova}
\author[2]{\\ Zinovi Rabinovich}
\author[1]{Alexandros A. Voudouris}

\affil[1]{Department of Computer Science, University of Oxford}
\affil[2]{School of Computer Science and Engineering, Nanyang Technological University}

\date{}

\begin{document}

\maketitle

\begin{abstract}
Complexity of voting manipulation is a prominent topic in computational
social choice. In this work, we consider a two-stage voting manipulation
scenario. First, a malicious party
(an attacker) attempts to manipulate the election outcome in favor of a preferred candidate
by changing the vote counts
in some of the voting districts.
Afterwards, another party (a defender), which cares
about the voters' wishes, demands
a recount in a subset of the manipulated districts, restoring their vote counts
to their original values.
We investigate the resulting Stackelberg game
for the case where votes are aggregated using two variants of the Plurality rule,
and obtain an almost complete picture of the complexity landscape,
both from the attacker's and from the defender's perspective.
\end{abstract}

\section{Introduction}\label{sec:intro}
Democratic societies use elections to select their leaders. However, in societies without
a strong democratic tradition, elections may be used as a way to legitimize the status quo:
voters are asked to cast their ballots, but the election authorities do not count these ballots
correctly, in order to produce an outcome that favors a specific candidate. There are multiple reports
of such cases in Russia\footnote{\url{https://reut.rs/2Gf2FD5}},
Congo\footnote{\url{https://on.ft.com/2SW7ggy}} and Colombia\footnote{\url{https://colombiareports.com/voting-fraud-in-colombia-how-elections-are-rigged/}},
as well as a number of other countries. Even when the election authorities are trustworthy,
election results may be corrupted by an external party, for instance,
by means of hacking electronic voting machines~\citep{SFD+14,HT15}.

There are several ways to counteract electoral fraud. One approach is to send observers to polling
stations, to ensure that only eligible voters participate in the elections and their ballots are counted
correctly. However, it may be infeasible for the party that wants to protect the elections
(the {\em defender}) to send observers to all polling stations. Consequently, the election manipulator
(the {\em attacker}) may observe which polling stations remain unprotected, and focus
their effort on these stations. Thus, under this approach the attacker benefits from the second-mover
advantage.

An alternative approach that the defender can explore is to request recounts in some
of the voting districts. While recounts cannot protect from all forms of attacks on election integrity
(e.g., a recount is of limited use if voters have been bribed to vote in a specific way,
or if the polling station has been burned down), they are feasible
in a range of settings and offer the defender the second-mover advantage. Indeed, there are several
examples where a recount changed the election outcome. For instance, in the 2008 United States
Senate election in Minnesota the Democratic candidate Al Franken won the seat after a recount revealed
that 953 absentee ballots were wrongly
rejected\footnote{\url{https://bit.ly/2S2PMxY}}, and in the 2004 race for governor
in Washington the Democratic candidate Gregoire was declared the winner after three
consecutive recounts\footnote{\url{https://bit.ly/2tnO4gG}}.

However, recounts can be costly. In Gregoire's case, the Democratic party paid \$730000 for a statewide
manual recount, and in the 2016 US Presidential Election the fee to initiate a recount in Wisconsin
was \$3.5 million. Thus, a party that would like to initiate a recount in order
to rectify the election results should allocate its budget carefully.
Of course, the attacker also incurs costs to carry out the fraud: local election officials
may need to be bribed or intimidated, and the more districts are corrupted, the higher is the risk
that the election results will not be accepted.

\begin{table*}[t]
\centering
\small
\renewcommand{\arraystretch}{1.4}
\setlength{\tabcolsep}{12pt}

\begin{tabular}{l  l  l  l}
\noalign{\hrule height 1pt}
&\bf {Plurality over Voters ($\PV$)} & \multicolumn{2}{l}{\bf Plurality over Districts ($\PD$)}  \\
&    & Unweighted & Weighted \\
\noalign{\hrule height 0.5pt}
\multirow{3}{*}{\sc Rec} & NP-c, Thm.~\ref{thm:PV-Rec-hardness} (i) \threecand & {P, Thm.~\ref{thm:PD-Rec-unweighted}} & NP-c, Thm.~\ref{thm:PD-Rec-hardness} (i) \threecand \\[-1mm]

& NP-c, Thm.~\ref{thm:PV-Rec-hardness} (ii) \unary & & NP-c, Thm.~\ref{thm:PD-Rec-hardness} (ii) \unary \\[-1mm]

& $O(n^{m+2})$, Thm.~\ref{thm:PV-Rec-poly} & & $O(n^{m+2})$, Thm.~\ref{thm:PD-Rec-poly} \\[2mm]

\multirow{2}{*}{\sc Man}  &  NP-h, {Thm.~\ref{thm:PV-Man-hardness} (i)} \threecand~\bdzero~\uncon & NP-c, {Thm.~\ref{thm:PD-Man-NPhardness2}} \unary & $\Sigma_2^P$-c, Thm.~\ref{thm:PD-Man-hardness} \threecand \\[-1mm]

& NP-h, Thm.~\ref{thm:PV-Man-hardness} (ii) \unary~\bdzero~\uncon  &  & NP-h, Thm.~\ref{thm:PD-Man-NPhardness1} \unary~\bdzero \\
\noalign{\hrule height 1pt}
\end{tabular}

\caption{Summary of our complexity results. {\sc Man} denotes the attacker's problem,
and {\sc Rec} denotes the defender's problem. 
Hardness results with {\unary}
hold even when the input is given in unary (the default is binary);
with {\threecand} hold even for three candidates;
with {\bdzero} hold even when the defender's budget is zero;
with {\uncon} hold even when the attacker can change as many votes as she wants in each district.
}
\label{fig:results}
\end{table*}

\paragraph{Our Contribution.}
In this paper we analyze the strategic game associated
with vote recounting. In our model, there are two players:
the attacker, who modifies some of the votes in order to make his preferred
candidate $p$ the election winner,
and the defender, who observes the attacker's actions and tries to restore the correct outcome
(or, more broadly, to ensure that a candidate who is better than $p$ wins the election)
by means of recounting some of the votes. We assume that the set of voters is partitioned
into electoral districts, and both the defender and the attacker
make their choices at the level of districts rather than individual votes.
The attacker selects a subset of at most $B_{\cal A}$ districts and changes the vote counts
in the selected districts, and the defender can then restore the vote counts in
at most $B_{\cal D}$ districts to their original values.
We assume that both players have full information about the true votes
and each other's budgets, and the defender can observe the attacker's actions.
While the full information assumption is not entirely realistic,
we note that in a district-based model both parties only need to know the vote counts in each district
rather than individual votes, and one can get fairly accurate district-level information
from independent polls.
Also, verifying whether the votes in a district have been tampered with is possible using 
risk-limiting audits~\cite{LS12,S16}.

For simplicity, we focus on the Plurality voting rule, where each voter
votes for a single candidate. We consider two implementations of this rule:
(1) Plurality over Voters, where districts are only used for the purpose of collecting the ballots
and the winner is selected among the candidates that receive the largest number of votes in total,
and
(2) Plurality over Districts, where each district selects a preferred candidate using the Plurality rule,
and the overall winner is chosen among the candidates supported by the largest number of districts;
we also consider a variant of the latter rule where districts have weights,
and the measure of a candidate's success is the total weight of districts that support her.
Both of these rules are widely used in practice. For example, Plurality over Voters
is commonly used in gubernatorial elections in the US, while Plurality
over Districts is used in the US Presidential elections.

We provide a detailed analysis of the computational complexity of the
algorithmic problems faced by the attacker and the defender.
Our main results are summarized in Table~\ref{fig:results}.
Briefly, assuming that the vote counts and the weights of the districts are specified
in binary, most of the problems we consider are computationally hard; however, the defender's problem appears
to be easier than that of the attacker, and we also get some tractability results for the former.
Towards the end of the paper,
we consider a variant of our model where the attacker is limited to only transferring votes
to his preferred candidate; we show that, while this assumption reduces the attacker's
ability to achieve his goals, it lowers the complexity of some of the problems we consider.

\paragraph{Related Work.}
There is a very substantial literature on voting manipulation and bribery; we point the readers
to the excellent surveys of \citet{CW16} and \citet{FR16}. In much of this literature
it is assumed that the malicious party can change some of the votes subject to
various constraints, and the challenge is to determine whether the attacker's task
is computationally feasible; there is no defender that can counteract the attacker's
actions.

While there is a number of papers that apply game-theoretic analysis
to the problem of voting manipulation, they typically consider interactions
between several manipulators, with possibly conflicting goals (e.g., see the recent book
by \citet{M18}), rather than a manipulator and a socially-minded actor.
An important exception, which is similar in spirit to our paper,
is the recent work of \citet{yin2018optimal},
who investigate a pre-emptive approach to protecting elections. In their model
the defender allocates resources
to guard some of the electoral districts, so that the votes there cannot be corrupted;
notably, in this model the defender has to commit to its strategy first, and the attacker
can observe the defender's actions before deciding on its response.
The leader-follower (defender-attacker) structure of this model
is in the spirit of a series of successful applications of Stackelberg
games to security resource allocation problems~\citep{tambe2011security}.
\citet{li2017protecting} analyze a variant of the model of Yin et al.~where the goal is to minimize
resource consumption, and \citet{chen2018protecting} study a similar scenario,
in which manipulation is achieved through bribing the voters.
The key difference between our work and the above papers is the action order of the players:
in all prior work on election protection that we are aware of the defender makes the first move.

\section{The Model}\label{sec:model}
We consider elections over a {\em candidate set} $C$, $|C|=m$.
There are $n$ {\em voters} who are partitioned into $k$
pairwise disjoint {\em districts} $D_1, \dots, D_k$, $k\le n$;
for each $i\in [k]$, let $n_i=|D_i|$.
For each $i\in [k]$,
district $D_i$ has a {\em weight $w_i$}, which is a positive integer;
we say that an election is {\em unweighted} if $w_i=1$ for all $i\in [k]$.
Each voter votes for a single candidate in $C$.
For each $i\in [k]$ and each $a\in C$ let $v_{ia}$ denote the
number of votes that candidate $a$ gets from voters in $D_i$;
we refer to the list $\vv = (v_{ia})_{i\in [k], a\in C}$ as the {\em vote profile}.

Let $\succ$ be a linear order over $C$; $a \succ b$ indicates that $a$ is favored over $b$.
We consider the following two voting rules, which take the vote profile
$\vv$ as their input.
\begin{itemize}
\item
\emph{Plurality over Voters} (PV).
We say that a candidate $a$ {\em beats} a candidate $b$ under PV
if $\sum_{i\in [k]}v_{ia} > \sum_{i\in [k]}v_{ib}$ 
or $\sum_{i\in [k]}v_{ia} = \sum_{i\in [k]}v_{ib}$ and $a\succ b$;
the winner is the candidate that beats all other candidates. 
Note that district weights $w_i$ are not relevant for this rule.
\item
\emph{Plurality over Districts} (PD).
For each $i\in [k]$ the winner $a_i$ in $D_i$
is chosen from the set $\arg\max_{a\in C} v_{ia}$, with ties broken according to $\succ$.
Then, for each $i\in [k]$, $a\in C$, we set $w_{ia} = w_i$ if $a=a_i$
and $w_{ia}=0$ otherwise.
We say that a candidate $a$ {\em beats} a candidate $b$ under PD
if $\sum_{i\in [k]}w_{ia} > \sum_{i\in [k]}w_{ib}$ 
or $\sum_{i\in [k]}w_{ia} = \sum_{i\in [k]}w_{ib}$ and $a\succ b$;
the winner is the candidate that beats all other candidates. 
\end{itemize}

For PV and PD, we define the {\em social welfare} of a candidate $a\in C$
as the total number of votes that $a$ gets and the total weight that $a$ gets, respectively:
$$
\SW^\PV(a) = \sum_{i \in [k]} v_{ia}, \quad
\SW^\PD(a) = \sum_{i \in [k]} w_{ia}.
$$
Hence, the winner under each voting rule is a candidate with the maximum social welfare.

We consider scenarios where an election may be manipulated by an {\em attacker}, who
wants to change the election result $a^*$ in favor of his preferred candidate $p\in C$.
The attacker has a budget $B_\calA \in [k]$, which means that he can manipulate
at most $B_\calA$ districts. For each $i\in [k]$, we are given an integer $\gamma_i$,
$0\le \gamma_i\le n_i$, which indicates how many votes the attacker can change
in district $i$ if he chooses to manipulate it.
Formally, a {\em manipulation}
is described by a set $M\subseteq [k]$, $|M|\le B_\calA$,
and a vote profile $\tildevv = (\tilde{v}_{ia})_{i\in [k], a\in C}$ such that
$\tilde{v}_{ia}=v_{ia}$ for all $i\not\in M$, $a\in C$, and for all $i\in [k]$
it holds that $\sum_{a\in C}\tilde{v}_{ia}=n_i$ and
$\sum_{a\in C}\max\{0, \tilde{v}_{ia}-v_{ia}\}\le \gamma_i$.

After the attack, a {\em defender} with budget $B_\calD \in \{0\}\cup[k]$ can demand a recount
in at most $B_\calD$ districts. Formally, a defender's strategy is a set $R\subseteq M$ with
$|R|\le B_\calD$; after the defender acts, the vote counts in all districts in $R$ are
restored to their original values, i.e., the resulting vote profile $\uu=(u_{ia})_{i\in [k], a\in C}$
satisfies $u_{ia}=v_{ia}$ for each $i\in R$, $a\in C$ and $u_{ia}=\tilde{v}_{ia}$
for each $i\in [k]\setminus R$, $a\in C$. Then the underlying voting rule $\calR\in\{\PV, \PD\}$
is applied to $\uu$ with ties broken according to $\succ$; let $a'$ denote the candidate selected
in this manner. The defender chooses her strategy $R$ so as to maximize $\SW^\calR(a')$,
breaking ties using $\succ$. 

We say that the attacker {\em wins}
if he has a strategy $(M, \tildevv)$ such that, once the
defender responds optimally, candidate $p$ is the winner in the resulting vote profile $\uu$;
otherwise we say that the attacker {\em loses}. We note that if
$B_\calD\ge B_\calA$, the defender can always ensure that $a'=a^*$, i.e., the winner at $\uu$
is the winner at the original vote profile $\vv$,
so in what follows we assume that the attacker's strategy satisfies $|M| > B_\calD$.

\begin{example}
Consider an election with five districts $D_1, \dots, D_5$ over a candidate set $C=\{a, b, p\}$,
where $p$ is the attacker's preferred candidate; suppose that ties are broken according
to the priority order $p \succ a \succ b$.
In each of $D_1$ and $D_2$ there are $7$ voters who vote for $a$, and in each of
$D_3$, $D_4$ and $D_5$ there are $3$ voters who vote for $b$. Suppose that $\gamma_i=n_i$
and $w_i=(n_i)^2$ for each $i\in [5]$, and $B_\calA=2$, $B_\calD=1$.

If the voting rule is $\PV$, then the attacker does not have a winning strategy.
Indeed, consider an attacker's strategy $(M, \tildevv)$. If $M\neq \{1, 2\}$,
the defender can set $R=M\cap\{1, 2\}$; in the recounted vote profile $a$ gets at least $14$
votes, so it is the election winner. If $M=\{1, 2\}$, the defender can set $R=\{1\}$:
in the recounted vote profile $p$ gets at most $7$ votes, while $b$ gets at least $9$ votes,
so the winner is $a$ or $b$ ($a$ can win if, e.g., the attacker chooses to transfer exactly $4$ votes
from $a$ to $p$ in $D_2$, in which case $a$ gets $10$ votes after the recount).
Note that even if the winner in $\uu$ is $b$ rather than $a$,
the defender still prefers recounting $D_1$ to no recounting:
even though she cannot restore the correct result, she prefers $b$ to $p$, since
$\SW^\PV(b)=9>0=\SW^\PV(p)$.

If the voting rule is $\PD$, then the attacker can win by choosing $M=\{1, 2\}$
and transferring a majority of votes from $a$ to $p$ in both districts.
Indeed, even if the defender demands a recount in one of these districts, $p$ still wins
the remaining district, leading to a vote weight of $49$ in the recounted profile.
Since $a$'s vote weight is $49$ and $b$'s vote weight is $27$, $p$
wins by the tie-breaking rule.
\hfill\qed
\end{example}

We assume that both the defender and the attacker have full information about the game.
Both parties know the true vote profile $\vv$, the parameters $w_i$ and $\gamma_i$
for each district $i\in [k]$ and each others' budgets.
Moreover, the defender observes the strategy $(M,\tildevv)$ of the attacker.

We can now define the following decision problems for each
$\calR \in\{\PV,\PD\}$:
\begin{itemize}
\item {\sc $\calR$-Man}:
Given a vote profile $\vv$, the attacker's preferred candidate $p$,
budgets $B_\calA$ and $B_\calD$, and district parameters $(w_i, \gamma_i)_{i\in [k]}$,
does the attacker have a winning strategy?

\item {\sc $\calR$-Rec}: Given a vote profile $\vv$,
a distorted vote profile $\tildevv$ with winner $b$,
a candidate $a \neq b$,
a budget $B_\calD$, and district weights $(w_i)_{i\in [k]}$,
can the defender recount the votes in at most $B_\calD$
districts so that $a$ gets elected?
\end{itemize}
We will also consider an optimization version of {\sc $\calR$-Rec},
where $c$ is not part of the input and
the goal is to maximize the social welfare of the eventual winner.

Unless specified otherwise, we assume that the vote counts $v_{ia}$
and the district weights $w_i$ are given in binary; we explicitly indicate
which of our hardness results still hold
if these numbers are given in unary.
All problems considered in this paper admit straightforward greedy algorithms
for $m=2$, so in what follows we focus on the case $m\ge 3$.
When the voting rule $\calR \in \{\PV,\PD\}$ is clear from context,
we write $\SW(a)$ instead of $\SW^\calR(a)$.

Next, we give formal definitions of the decision problems that are used throughout the paper
to show hardness of {\sc $\calR$-Rec} and {\sc $\calR$-Man} for $\calR \in \{\PV,\PD\}$,
under various constraints.

\begin{definition}[{\sc Subset Sum}]\label{def:SS}
An instance of {\sc Subset Sum} is given by a multiset $X$ of integers.
It is a yes-instance if there exists a non-empty subset $X' \subseteq X$
such that $\sum_{x \in X'} x = 0$, and a no-instance otherwise.
\end{definition}

\begin{definition}[{\sc Exact Cover By 3-Sets} (X3C)]\label{def:X3C}
An instance of X3C is given by a set $E$ of size $3\ell$ and a collection $\calS$
of $3$-element subsets of $E$.
It is a yes-instance if there exists a sub-collection $\calQ\subseteq \calS$ of size $\ell$
such that $\cup_{S\in \calQ}S=E$, and a no-instance otherwise.
\end{definition}

\begin{definition}[{\sc Independent Set}]\label{def:IS}
An instance of {\sc Independent Set} is a graph $G=(V,E)$ and an integer $\ell$.
It is a yes-instance if there exists a subset $V'\subseteq V$
of size $\ell$ that forms an independent set, i.e.,
$\{a, b\}\not\in E$ for all $a, b\in V'$,
and a no-instance otherwise.
\end{definition}

\begin{definition}[{\sc Partition}]\label{def:Partition}
An instance of {\sc Partition} is given by a multiset $X$ of positive integers.
It is a yes-instance if there exists a subset $X' \subseteq X$
such that $\sum_{x \in X'} x = \frac{1}{2} \sum_{x \in X} x$, and a no-instance otherwise.
\end{definition}

All of these problems are NP-complete~\citep{gj79}.
However, {\sc Subset Sum} and {\sc Partition} are NP-hard only when the input is given in binary;
for unary input, these problems can be solved in time polynomial in the size of the input.

\section{Plurality over Voters}\label{sec:pv}
In this section we focus on Plurality over Voters. We first take the perspective of the defender,
and then the perspective of the attacker.

Unfortunately, the defender's problem turns out to be computationally hard, even if
there are only three candidates or if the input vote counts are given in unary.

\begin{theorem} \label{thm:PV-Rec-hardness}
{\sc $\PV$-Rec} is {\em NP}-complete even when
\begin{itemize}
  \item[(i)] $m=3$, or
  \item[(ii)] the input vote profile is given in unary.
\end{itemize}
\end{theorem}

\begin{proof}
This problem is clearly in NP. We give separate hardness proofs
for the case $m=3$ (part (i))
and for the case where the input is given in unary (part (ii)).

\paragraph{Part (i).}
To prove that {\sc $\PV$-Rec} is NP-hard for $m=3$,
we provide a reduction from {\sc Subset Sum}; see Definition~\ref{def:SS}.

Given an instance $X$ of {\sc Subset Sum} with $|X|=\ell$,
we construct an instance of $\PV$-{\sc Rec} as follows. Without loss of generality, we assume
that $x \neq 0$ for every $x \in X$ and $\sum_{x\in X}x> 0$,
and let $X^+=\{x\in X: x>0\}$, $X^-=\{x\in X: x<0\}$,
$y = \sum_{x\in X} 2|x|$.
We set $C=\{a, b, p\}$,
where $p$ is the attacker's preferred candidate. In what follows, we describe
each district $D_i$ by a tuple $(v_{ia}, v_{ib}, v_{ip})$.
There are $n = 12y\ell + 1$ voters distributed over $\ell+3$
districts, which are further partitioned into two sets $I_1$ and $I_2$ as follows:
\begin{itemize}
\item
For each $x \in X^+$ there is a district in $I_1$
with votes $(0, 2x, 0)$, which are distorted to $(0, 0, 2x)$, and
for each $x \in X^-$ there is a district in $I_1$
with votes $(0, 0, -2x)$, which are distorted to $(0, -2x, 0)$.
Note that $|I_1|=\ell$.
\item
$I_2$ contains three districts with votes $(y+1,0,0)$, $(0,y-\sum_{x\in X^+}2x,0)$,
and $(0,0,y + \sum_{x\in X^-}2x)$, respectively. The votes in these districts are not distorted.
\end{itemize}
Finally, $B_\calD = \ell - 1$.

Before the manipulation, $a$ gets $y+1$ votes and $b$ and $p$ get $y$ votes each.
After the manipulation, $a$ gets $y+1$ votes, $b$ gets
$y-\sum_{x\in X}2x$ and $p$ gets $y+\sum_{x\in X}2x$ votes;
thus, by our assumption that $\sum_{x\in X}x > 0$, candidate $p$ is the winner
in the manipulated profile. The goal is to restore the true winner $a$. 

Now, assume that there exists a subset $X' \subseteq X$ with $|X'| \geq 1$ such that $\sum_{x \in X'}x = 0$.
Then, by recounting the $\ell - |X'|$ districts of $I_1$ that correspond to the integers in
$X\setminus X'$, the defender can ensure that both $b$ and $p$ get $y$ votes. Since $a$ always gets
$y+1$ votes from the non-manipulated districts, she is successfully restored as the winner.

Conversely, assume that there is no non-empty subset $X' \subseteq X$ such that $\sum_{x \in X'} x = 0$. Then,
since the votes of $b$ and $p$ always add up to exactly $2y$, and each of them gets
an even number of votes from each district, one of them must get at least $y+2$
votes. Therefore, $a$ cannot be restored as the winner.

\paragraph{Part (ii).}
We give a reduction from {\sc Exact Cover By 3-Sets (X3C)}; see Definition~\ref{def:X3C}.
Given an instance of X3C,
we construct the following {\sc $\PV$-Rec} instance.
Without loss of generality, we assume that
$\cup_{S\in \calS}S= E$, and let $s=|\calS|$.
\begin{itemize}
\item Let $C = \{ j_e: e\in E \} \cup \{a,b\}$, $|C|=3\ell+2$.

\item For each subset $S\in \calS$, there is a district $D_S$,
where $a$ gets $2$ votes, $b$ gets $6$ votes, for each $e\notin S$
candidate $j_e$ gets $2$ votes, and for each $e\in S$ candidate $j_e$
gets $0$ votes. The attacker distorts the votes in $D_S$
by transferring two votes from $b$ to each candidate $j_e$ with $e\in S$,
so that in the distorted profile
$b$ gets $0$ votes in $D_S$ and every other candidate
gets $2$ votes in $D_S$.

\item There is a district $D_0$ where
$a$ receives $6\ell s$ votes, $b$ receives $0$ votes and
for every $e \in E$ candidate $j_e$ receives $6\ell s +1$ votes;
the votes in this district are not distorted.

\item The budget of the defender is $B_\calD=\ell$.
\end{itemize}
Candidate $a$ is the true winner with $2s+ 6\ell s$ votes,
compared to the $6s$ votes of $b$ and the
$2|\{S \in \calS: e \notin S\}| + 6\ell s+1 \leq 2 s+6\ell s-1$ votes of $j_e$ for every $e \in E$.
In the distorted profile $\tildevv$ candidate $a$ gets $2s+6\ell s$ votes, candidate $b$
gets $0$ votes, and each candidate in $C\setminus\{a, b\}$ gets
$2s+ 6\ell s+1$ votes.

Recounting a district $D_S$ reduces by $2$ the votes of each candidate $j_e$ such
that $e \in S$, leading to $a$ getting more votes than these candidates; $b$ cannot get more
than $6s$ votes no matter what the defender does.
Therefore, $a$ can be restored as the winner by recounting
$\ell$ districts if and only if $E$ can be covered by $\ell$ sets from $\calS$.
\end{proof}

If the number of candidates is bounded by a constant and the input
is given in unary, an optimal set of districts to recount can be identified in
time polynomial in the input size by means of dynamic programming.

\begin{theorem}\label{thm:PV-Rec-poly}
{\sc $\PV$-Rec} can be solved in time $O(k \cdot B_\calD \cdot (n+1)^m)$.
\end{theorem}

\begin{proof}
Consider an instance of {\sc $\PV$-Rec} with a candidate set $C$, $|C|=m$,
and $n$ voters that are distributed over $k$ districts.
For each $i\in [k]$, let
$\vv_i = (v_{ia})_{a \in C}$ and
$\tildevv_i = (\tilde{v}_{ia})_{a \in C}$
denote, respectively, the true and distorted votes in district $i$.
Let $B_\calD$ be the budget of the defender.

We present a dynamic programming algorithm that given a candidate $c\in C$,
decides whether $c$ can be made the election winner by recounting at most $B_\calD$ districts.
Our algorithm fills out a table $T$ containing entries of the form $T(i, \ell, \uu)$,
for each $i \in \{0,1,\dots, k\}$, $\ell \in \{0,1,\dots,B_\calD\}$,
and $\uu= (u_a)_{a \in C} \in \{0, \dots, n\}^m$;
thus, $|T|=O(k \cdot B_\calD \cdot (n+1)^m)$.
We define $T(i, \ell, \uu) = \true$ if
we can recount at most $\ell$ of the first $i$ districts so
that the vote count of candidate $a$ equals $u_a$ for each $a \in C$;
otherwise we define $T(i, \ell, \uu) = \false$.
There exists a recounting strategy that restores $c$ if and only if there
exists a $\uu$ such that $T(k, B_\calD, \uu) = \true$,
$u_c \ge u_a$ for all $a \in C$,
and for all $a\in C\setminus\{c\}$ such that $u_c=u_a$
the tie-breaking rule favors $c$ over $a$.

For each $a\in C$, let $\tilde{u}_a = \sum_{i \in [k]} \tilde{v}_{ia}$
be the number of votes that candidate $a$ gets after manipulation,
and let $\tildeuu=(\tilde{u}_a)_{a\in C}$.
We fill out $T$ according to the following rule:
\begin{align*}
T(i, \ell, \uu ) =
\begin{cases}
\true, ~\quad \text{if } \uu = \tildeuu \\
\false, \quad  \text{if $i=0$ or $\ell=0$, and $\uu \neq \tildeuu$}  \\
T(i-1, \ell, \uu) \vee (\uu-\vv_i+\tildevv_i\in \{0, \dots, n\}^m~\mbox{and}\\
 \qquad T \left( i-1, \ell-1, \uu - \vv_i + \tildevv_i \right)), \ \text{otherwise.}
\end{cases}
\end{align*}
This completes the proof.
\end{proof}

We obtain similar hardness results for the attacker's problem. However,
it is not clear if {\sc $\PV$-Man} is in NP. Indeed,
it may belong to a higher level of the polynomial hierarchy: it is not hard to see
that {\sc $\PV$-Man} is in $\Sigma_2^P$, and it is plausible that
this problem is hard for this complexity class.

\begin{theorem}\label{thm:PV-Man-hardness}
{\sc $\PV$-Man} is {\em NP}-hard even when $B_\calD=0$,  $\gamma_i=n_i$ for all $i\in [k]$ and
\begin{itemize}
  \item[(i)] $m=3$, or
  \item[(ii)] the input vote profile is given in unary.
\end{itemize}
\end{theorem}

\begin{proof}
We prove the two claims separately.

\paragraph{Part (i).}
To prove that {\sc $\PV$-Man} is NP-hard for $m=3$,
we provide a reduction from {\sc Subset Sum};
see Definition~\ref{def:SS}.

Given an instance $X$ of {\sc Subset Sum} with $|X|=\ell$,
we construct an instance of PV-{\sc Man} as follows. We can assume without loss of generality
that $\ell\ge 2$ and $x\neq 0$ for every $x\in X$, and let $y=\max_{x\in X}2|x|$; by our assumptions, $y\ge 2$.
We set $C=\{a, b, p\}$,
where $p$ is the attacker's preferred candidate. In what follows, we describe
each district $D_i$ by a tuple $(v_{ia}, v_{ib}, v_{ip})$.
There are $n = 12y\ell + 1$ voters distributed over
$k=4\ell+2$ districts, which are further partitioned into four sets $I_1, I_2, I_3, I_4$ as follows:
\begin{itemize}
\item
For each $x \in X$ there is a district in $I_1$
with votes $(2y+4x, 2y - 4x, 0)$. Thus, $|I_1|=\ell$.
\item
Set $I_2$ consists of $\ell-1$ districts with votes $(2y, 2y, 0)$ in each district.
\item
For each $x \in X$ there are two districts in $I_3$
with votes $(y-2x, y + 2x, 0)$. Thus, $|I_3|=2\ell$.

\item Set $I_4$ consists of three districts with
votes $(y,y,0)$, $(y,y,0)$, and $(0,0,1)$.
\end{itemize}
We set $B_\calA = \ell$, $B_\calD = 0$ and $\gamma_i=n_i$ for each $i\in [k]$.

We have $\SW(a)=\SW(b)=6y\ell$ and $\SW(p)=1$.
Hence, the true winner is $a$ or $b$, depending on the tie-breaking rule.
We claim that the attacker can make $p$ the winner if and only if
there exists a non-empty subset $X' \subseteq X$ such that $\sum_{x \in X'} x = 0$.

To see this, assume first that there exists a subset $X' \subseteq X$ such that
$|X'|\ge 1$ and $\sum_{x \in X'} x = 0$.
Then the attacker can distort the votes in the $|X'|$ districts of $I_1$
corresponding to the elements of $X'$, and in arbitrary $\ell-|X'|$ districts of $I_2$,
by transferring all votes to $p$ in each of these districts. In the resulting election,
$p$ gets $4y\ell+1$ votes, while
$a$ and $b$ get $4y\ell$ votes each, so $p$ becomes the winner.

Conversely, suppose that the attacker has a successful manipulation $(M, \tildevv)$ with $|M|\le \ell$.
For each $c\in C$, let $s_c$ denote the number of votes that $c$ receives in $\tildevv$.
For $p$ to be the winner in $\tildevv$, it must hold that $s_p\ge n/3 = (12y\ell+1)/3$;
since $s_p$ is an integer and $\SW(p)=1$, this means that the manipulation transfers
at least $4y\ell$ votes to $p$.
On the other hand, in every district there are at most $4y$ voters who vote for $a$ or $b$, so $p$ can gain
at most $4y\ell$ votes from the manipulation. It follows that $s_p=4y\ell+1$,
$s_a+s_b=8y\ell$.
If these $8y\ell$ votes are not split evenly between $a$ and $b$,
at least one of these candidates would get strictly more than $4y\ell$ points; since each district
allocates an even number of votes to both $a$ and $b$, this further means that one of them would get
at least $4y\ell+2$ votes, a contradiction with $p$ being the winner at $\tildevv$.
Thus, it must be the case that $s_a = s_b = 4y\ell$.

Further, $s_p=4y\ell+1$, $|M|=\ell$ implies that $M\subseteq I_1\cup I_2$ and
$M \cap I_1 \neq \varnothing$. Moreover, we have
$\tilde{v}_{ia} = \tilde{v}_{ib} = 0$ for every district $i\in M$.
Hence,
$$
s_a = 4y\ell - 4\sum_{i\in M\cap I_1} x_i, \qquad
s_b = 4y\ell + 4\sum_{i\in M\cap I_1} x_i,
$$
where $x_i$ is the integer in $X$ that corresponds to district $D_i$.
Thus, $\sum_{i\in M\cap I_1} x_i = 0$,
and hence $X' = \{x_i\in X: i\in M\cap I_1\}$
is a witness that $X$ is a yes-instance of {\sc Subset Sum}.

\paragraph{Part (ii).}
To prove that {\sc $\PV$-Man} is NP-hard when the input is given in unary,
we provide a reduction from X3C; see Definition~\ref{def:X3C}.

Given an instance $\langle E, \mathcal{S} \rangle$ of X3C with $|E|=3\ell$, $|\calS|=s$,
we construct an instance of {\sc $\PV$-Man} as follows.
We set $C = \{j_e : e \in E \} \cup \{p\}$, where $p$ is the attacker's preferred candidate.
The districts are partitioned into three sets $I_1, I_2, I_3$:
\begin{itemize}
\item
For each subset $S \in \calS$ the set $I_1$ contains a district $D_S$.
In this district each candidate $j_e$ such that $e \in S$ gets $3\ell$
votes, and all other candidates get no votes. Thus, $|D_S| = 9 \ell$.

\item
For each element $e\in E$, the set $I_2$
contains $3\ell s + 9\ell^2 - 3\ell\cdot|\{S \in \calS: e\in S\}|$ districts;
each of these districts consists of a single voter who votes for $j_e$.

\item The set $I_3$ contains a single district $D^*$ that consists of
$3\ell s -2\ell$ voters who vote for $p$.
\end{itemize}
We set $B_\calA = \ell$, $B_\calD = 0$ and $\gamma_i = n_i$ for all $i\in [k]$.

We have $\SW(j_e) = 3\ell s + 9\ell^2$ for all $e \in E$
and $\SW(p) = 3\ell s- 2\ell$. Hence, the true winner is the candidate in $C\setminus\{p\}$
who is favored by the tie-breaking rule.
We show that the attacker is able to make $p$ the winner if and only if
$E$ admits an exact cover by sets from $\calS$.

Suppose that $\calQ \subseteq \calS$ is an exact cover for $E$; note that $|\calQ| = \ell$.
The attacker can manipulate the $\ell$ districts in $I_1$ that correspond to sets in $\calQ$
by reassigning all the $9 \ell$ votes in each of them to $p$.
In the resulting election, $p$ gets $3\ell s + 9 \ell^2 - 2\ell$ votes,
while every other candidate $j_e$ gets $3 \ell s + 9 \ell^2 - 3 \ell$ votes,
as every $e$ is covered by exactly one set in $\calQ$.

Conversely, suppose the attacker has a successful manipulation $(M, \tildevv)$ with $|M| \le \ell$.
For each $c\in C$, let $s_c$ denote the number of votes that $c$ receives in $\tildevv$.
As $p$ can gain at most $9\ell$ votes for each district in $M$,
we have $s_p \le 3 \ell s + 9 \ell^2 - 2\ell$.
Let $\calQ=\{S\in \calS: D_S~\text{ is manipulated}\}$; note that $|\calQ|\le\ell$.
We claim that $\calQ$ is a cover for $E$. Indeed,
if for some $e \in E$ no district in $\{D_S : e \in S\}$ is manipulated,
the manipulation lowers the score of  $j_e$ by at most $\ell$, so
$s_{j_e} \ge 3 \ell s + 9 \ell^2 - \ell > s_p$, a contradiction.
\end{proof}

In the hardness reductions in the proof of Theorem~\ref{thm:PV-Man-hardness}
the defender's budget is $0$. This indicates that the attacker's problem remains
NP-hard even if the defender is known to use a heuristic (e.g., a greedy algorithm)
to compute her response.

We remark that {\sc $\PV$-Rec} and {\sc $\PV$-Man} with $B_\calD=0$ are very similar in spirit
to combinatorial (shift) bribery~\citep{BFNT16}. In both models, a budget-constrained
agent needs to select a set of vote-changing actions, with each action affecting
a group of voters. However, there are a few technical differences between the models.
For instance, in our model different actions are associated with non-overlapping groups of voters,
which is not the case in combinatorial shift bribery. On the other hand, in shift
bribery under the Plurality rule votes can only be transferred to/from the manipulator's
preferred candidate $p$, while our model does not impose this constraint
(see, however, Section~\ref{sec:regular}). Consequently, it appears that
the technical results in our paper cannot be derived from known results for combinatorial shift bribery.


\section{Plurality over Districts}\label{sec:pd}
In this section we study Plurality over Districts.
For the defender's problem, we can replicate
the results we obtain for Plurality over Voters,
by using similar techniques.

\begin{theorem}\label{thm:PD-Rec-hardness}
{\sc $\PD$-Rec} is {\em NP}-complete even when
\begin{itemize}
  \item[(i)] $m=3$, or
  \item[(ii)] the input vote profile and district weights are given in unary.
\end{itemize}
\end{theorem}

\begin{proof}
This problem is clearly in NP. We give separate hardness proofs
for the case $m=3$ (part (i))
and for the case where the input is given in unary (part (ii)).

\paragraph{Part~(i).}
We use the same reduction as in the proof of the first part of Theorem~\ref{thm:PV-Rec-hardness}.
An important feature of this reduction is that all voters in each district vote for the same candidate.
Thus, if we set the weight of each district to be equal to the number of voters therein,
the proof goes through without change.

\paragraph{Part~(ii).}
We provide a reduction from {\sc Independent Set}; see Definition~\ref{def:IS}.
Given an instance $\langle G, \ell\rangle$ of {\sc Independent Set}, where $G=(V, E)$,
we construct an instance of {\sc $\PD$-Rec} as follows.
Let $\nu=|V|$, $\mu=|E|$; we can assume without loss of generality that $\mu\ge 1$.
We set $C = \{j_u : u \in V \} \cup \{j_e : e \in E\} \cup \{a, p\}$, where $p$ is the attacker's
preferred candidate; thus, $|C| = \nu + \mu + 2$.
We create the following districts. For our argument, the district sizes and the values of $\gamma_i$
do not matter; for concretness, we assume that each district
consists of a single voter, whose vote can be changed by the manipulator.
\begin{itemize}
\item
For each edge $e=\{x, y\} \in E$,
there are two districts $D_{e,x}$ and $D_{e,y}$ with weight $2$ each. In each such district $D_{e, u}$
the winner before manipulation is $j_e$, and the winner after manipulation is $j_u$.
\item
For each node $u \in V$, there is a district $D_u$ with weight $2\mu$;
in this district
the winner before manipulation is $j_u$, and the winner after manipulation is $p$.
\item
There is a set $I$ of $2(\nu+\mu)+1$ districts with weight $\frac{2}{2(\nu+\mu)+1}$
each\footnote{For
convenience, we use fractional weights. We can turn all weight into integers,
by multiplying them by $2(\nu+\mu)+1$.};
in each such district the winner before manipulation is $a$,
and the winner after manipulation is $p$.
\item
There is a district of weight $2(\nu-\ell)\mu + 3$ with winner $a$;
this district is not manipulated.
\item
For each $e\in E$, there is a district of weight
$2(\nu-\ell)\mu$ with winner $j_e$;
this district is not manipulated.
\item
For each $u\in V$, there is a district of weight $2(\nu-\ell)\mu -2\mu+2$
with winner $j_u$; this district is not manipulated.
\end{itemize}
The budget of the defender is $B_\calD = \nu + \mu$.
The candidates' weights before and after manipulation are given in the following table:
\begin{center}
\begin{tabular}{c l l}
\noalign{\hrule height 1pt}
	  & true weight & distorted weight \\
\noalign{\hrule height 1pt}
$a$ & 		$2(\nu-\ell)\mu + 5$ 	& $2(\nu-\ell)\mu+3$ \\
$p$ 		& $0$ 			& $2\nu\mu+2$ \\
$j_e$, $e\in E$ & $2(\nu-\ell)\mu+4$  	& $2(\nu-\ell)\mu$ \\
$j_u$, $u\in V$ & $2(\nu-\ell)\mu+2$	& $\leq 2(\nu-\ell)\mu+2$ \\
\noalign{\hrule height 1pt}
\end{tabular}
\end{center}
Hence, the true winner is candidate $a$ and the winner after manipulation is $p$.

If $V'\subseteq V$ is an independent set of size $\ell$ in $G$, the defender
can proceed as follows. For each $u\in V'$, she demands a recount in $D_u$
and in each district $D_{e, u}$ such that $e$ is incident to $u$. Since $V'$
forms an independent set, this requires recounting at most $\nu+\mu$ districts.
Moreover, after the recount the weight of $p$ is $2(\nu-\ell)\mu+2$,
the weight of $a$ is $2(\nu-\ell)\mu+3$,
the weight of each candidate $j_u$ such that $u\in V'$ is $2(\nu-\ell)\mu+2$,
the weight of each candidate $j_u$ such that $u\in V\setminus V'$ is at most $2(\nu-\ell)\mu+2$,
and the weight of each candidate $j_e$ such that $e\in E$ is at most $2(\nu-\ell)\mu+2$.
Thus, this recounting strategy successfully restores $a$ as the election winner.

Conversely, suppose that the defender has a recounting strategy $R$ that results in making $a$
the election winner. Since $|R|\le B_\calD$, at most $\nu+\mu$ districts in $I$
can be recounted, so $a$'s weight after the recount is at most
$2(\nu-\ell)+3+\frac{2(\nu+\mu)}{2(\nu+\mu)+1} < 2(\nu-\ell)+4$.
Now, if $R$ contains at most $\ell-1$ districts in $\{D_u : u\in V\}$,
then $p$'s weight after the recount is at least $2(\nu-\ell+1)\mu+2 \ge 2(\nu-\ell)+4$,
a contradiction with $a$ becoming the winner after the recount.
Hence, $R$ contains at least $\ell$ districts in $\{D_u : u\in V\}$;
let $V'$ be the subset of nodes corresponding to these districts.
We claim that $V'$ forms an independent set in $G$.

Indeed, consider a node $u\in V'$. If the defender does not recount some district
$D_{e, u}$ such that $u$ is incident to $e$ then after the recount the weight
of $j_u$ is at least $2(\nu-\ell)\mu+4$, a contradiction with $a$ becoming the winner
after the recount. Thus $D_{e, u}$ is necessarily recounted. Now, suppose
that $e = \{x, y\}\in E$ for some $x, y\in V$. We have just argued that
both $D_{e, x}$ and $D_{e, y}$ have to be recounted. But this means that the score of
$j_e$ is at least $2(\nu-\ell)\mu+4$ after the recount, a contradiction again.
Thus, $V'$ is an independent set.
\end{proof}

\begin{theorem}\label{thm:PD-Rec-poly}
{\sc $\PD$-Rec} can be solved in time $O(k \cdot B_\calD \cdot (n+1)^m)$.
\end{theorem}

\begin{proof}
The algorithm is a simple adaptation of the dynamic program
presented in the proof of Theorem~\ref{thm:PV-Rec-poly}.
\end{proof}

We also obtain a positive result that does not have an analogue in the PV setting; if all districts
have the same weight, the recounting problem can be solved efficiently.

\begin{theorem}\label{thm:PD-Rec-unweighted}
{\sc $\PD$-Rec} can be solved in polynomial time if $w_i=1$ for all $i \in [k]$.
\end{theorem}
\begin{proof}
We reduce our problem to nonuniform bribery~\citep{F08}.
An instance of nonuniform bribery under the Plurality rule is given by a set of voters
and a set of candidates; for each voter $i$ and each candidate $c$ there
is a price $\pi_{ic}$ for making voter $i$ vote for $c$,
and the briber's goal is to make her preferred candidate the Plurality winner\footnote{\citet{F08}
assumes that ties are broken in favor of the briber,
but his results extend to lexicographic tie-breaking.}
while staying within a budget $B$. This problem is known to be in P~\citep{F08}.
To reduce {\sc $\PD$-Rec} to nonuniform bribery,
we map each district $D_i$ to a single voter $i$;
if the true winner in $D_i$ is $x$,
and in the distorted profile the winner in $D_i$ is $y$, we set $\pi_{iy}=0$,
$\pi_{iz}=+\infty$ for $z\in C\setminus\{x, y\}$, and if $x\neq y$
(i.e., if the attacker has changed the outcome in $D_i$), we set $\pi_{ix}=1$.
Then for any candidate $c\in C$ it holds that in {\sc $\PD$-Rec}
the defender can make $c$ win by recounting at most $B_\calD$ districts
if and only if in our instance of nonuniform bribery the briber can make $c$ win
by spending at most $B_\calD$.
\end{proof}

We now consider the attacker's problem. It turns out that for the $\PD$ rule
we can obtain a stronger hardness result than for $\PV$:
we will now argue that when weights and vote counts are given in binary,
{\sc $\PD$-Man} is $\Sigma_2^P$-complete even for $m=3$.
Our reduction uses a variant of the {\sc Subset Sum}
problem, which we term {\sc Sub-Subset Sum} (SSS); this problem may be of independent interest.

\begin{definition}[{\sc Sub-Subset Sum}]\label{def:SSS}
An instance of {\sc Sub-Subset Sum} is a set $X\subseteq {\mathbb Z}$ and a positive integer $\ell$.
It is a yes-instance if there is a subset $X' \subseteq X$ with $|X'|=\ell$
such that $\sum_{x \in X''} x \neq 0$
for every non-empty subset $X'' \subseteq X'$, and a no-instance otherwise.
\end{definition}

\noindent
Our proof proceeds by establishing that SSS is $\Sigma_2^P$-complete (Lemma~\ref{lem:SSS-hardness};
the proof can be found in the appendix),
and then reducing this problem to {\sc $\PD$-Man}.

\begin{lemma}\label{lem:SSS-hardness}
SSS is $\Sigma_2^P$-complete.
\end{lemma}

\begin{theorem}
\label{thm:PD-Man-hardness}
{\sc $\PD$-Man} is $\Sigma_2^P$-complete, even when $m=3$.
\end{theorem}

\begin{proof}
Clearly, {\sc $\PD$-Man} is in $\Sigma_2^P$. To prove hardness, we reduce from SSS.
Given an instance $\langle X,\ell \rangle$ of SSS, we construct an instance of {\sc $\PD$-Man} with three candidates $\{a,b,p\}$.
Let $X^+ = \{x \in X: x > 0\}$ and $X^- = X \setminus X^+$.
Set $y = \sum_{x\in X} 3|x|$.
In what follows we describe the votes in each district $D_i$ as a list $(v_{ia},v_{ib},v_{ip})$.
The districts are partitioned into three sets $I_1$, $I_2$ and $I_3$:
\begin{itemize}
\item $I_1$ has a district with votes $(0, 3x, 0)$ for each $x \in X^+$,
and a district with votes $(0, 0, -3x)$ for each $x \in X^-$.

\item $I_2$ consists of a single district with votes $(0, y+3, 0)$.

\item $I_3$ consists of three districts with votes $(2y+5,0,0)$,
$(0,y-\sum_{x\in X^+}3x,0)$, and $(0,0, 2y + 4 + \sum_{x\in X^-}3x)$.
\end{itemize}
For every district $D_i$ we set $w_i=n_i$.
The attacker is allowed to change all votes in each district in $I_1$ and $I_2$,
but none in $I_3$. Finally, let $B_\calA = \ell + 1$ and $B_\calD = \ell$.
The true winner in this profile is candidate $a$ with weight $2y+5$, compared to the
weight $2y+3$ of $b$ and $2y+4$ of $p$.

Given a set of integers $Y\subseteq X$, let $I_1(Y)$ be the corresponding set of districts in $I_1$.
Assume that there is a subset $X' \subseteq X$ with $|X'|=\ell$ such that
no $X'' \subseteq X'$ has sum equal to $0$. The attacker can then exchange the weights of $b$ and
$p$ in the districts in $I_1(X')$ and the district in $I_2$.
This way, $p$ becomes the winner with weight $3y + 7 + \sum_{x \in X'} 3x  \ge 2y + 7$,
compared to the weight $2y+5$ of $a$ and the weight $y - \sum_{x \in X'} 3x \le 2y$ of $b$.

Since $\SW(p)>\SW(b)$, to defeat the attacker,
the defender needs to restore $a$ as the winner.
To this end, she must recount the district in $I_2$,
as otherwise $p$'s weight will remain
at least $2y+7$. Hence she can recount at most $\ell-1$
manipulated districts in $I_1$. Let the set of non-recounted districts in $I_1$
be $I_1(X'')$ for some $X''\subseteq X'$; note that $X''\neq\varnothing$,
so by assumption, $\sum_{x \in X''} x \neq 0$. Then, the weight of $b$ is
$2y + 3  - \sum_{x \in X''} 3x$
and the weight of $p$ is
$2y + 4 + \sum_{x \in X''} 3x$.
At least one of these numbers is greater than or equal to $2y+6$;
thus, $a$ cannot be restored as the winner.

Conversely, suppose that for every subset $X' \subseteq X$ of size $\ell$ there exists a non-empty
$X''\subseteq X'$ such that $\sum_{x\in X''} x = 0$. Then, the attacker cannot win.
Indeed, let $M$ be the set of manipulated districts.
If a district is changed in favor of $a$, the defender can recount all other districts in $M$.
On the other hand,
if all districts in $M$ are won by $b$ or $p$, the defender can identify a non-empty subset
of $M\cap I_1$
such that the corresponding integers sum up to $0$, and request a recount
of all other districts in $M$.
Such a recount recovers the correct weights of $b$ and $p$, and $a$
is restored as the winner.
\end{proof}

We conjecture that {\sc $\PD$-Man} remains $\Sigma_2^P$-complete
when the input is given in unary; however, for this setting we are
only able to prove that this problem is NP-hard.

\begin{theorem}
\label{thm:PD-Man-NPhardness1}
  {\sc $\PD$-Man} is {\em NP}-hard, even when $B_\calD =0$ and the input vote profile
       and district weights are given in unary.
\end{theorem}
\begin{proof}
To show that {\sc $\PD$-Man} is NP-hard even
when the input votes and district weights are given in unary,
we provide a reduction from {\sc Independent Set}; see Definition~\ref{def:IS}.

Given an instance $\langle G, \ell\rangle$ of {\sc Independent Set}
with $G=(V, E)$,
we construct the following instance of {\sc $\PD$-Man}.
Let $\nu=|V|$, $\mu=|E|$.
We set $C = \{j_u : u \in V \} \cup \{j_e : e \in E\} \cup \{a, p\}$, where $p$ is the attacker's
preferred candidate; thus, $|C| = \nu + \mu + 2$.
Then, we create the following districts;
the weight of each district is equal to the number of voters therein.
\begin{itemize}
\item
    For every edge $e=\{x,y\} \in E$, we create two districts $D_{e,x}$ and $D_{e,y}$
    with $5$ voters each; thus, $w_{e,x} = w_{e,y} = 5$.
    In each such district $D_{e,u}$ there are two voters who vote for $j_e$ and
    three voters who vote for $j_u$.
    We set $\gamma_{e, u}=1$; thus, the attacker can change the winner in this district
    from $j_u$ to $j_e$.
\item
    For every node $u\in V$, we create a district $D_u$ with $5\mu$ voters; thus, $w_u = 5\mu$.
    In each such district there are $2\mu$ voters who vote for $j_u$ and $3\mu$ voters who vote for $a$.
    We set $\gamma_u = \mu$; thus, the attacker can change the winner in this district from $a$
    to $j_u$.
\item
    There are also some districts that cannot be manipulated (i.e., $\gamma=0$).
    We specify the weights and the winners of these districts.
\begin{itemize}
\item For each $e\in E$, there is a district with weight $5\mu(\nu-\ell) - 5$ and winner $j_e$.
\item For each $u\in V$, there is a district with weight $5\mu(\nu-\ell-1)$ and winner $j_u$.
\item Finally, there is a district with weight $5\mu(\nu-\ell)+1$ and winner $p$.
\end{itemize}
\end{itemize}
The budgets are $B_\calA = \nu + \mu$ and $B_\calD = 0$.

We have $\SW(a) = 5\mu\nu$, $\SW(p) = 5\mu(\nu-\ell)+1$, $\SW(j_e) = 5\mu(\nu-\ell)-5$ for each $e \in E$,
and $\SW(j_u) = 5\mu(\nu-\ell - 1) + 5 \, |\{e\in E: u\in e \}| \leq 5\mu(\nu-\ell)$ for each $u \in V$.
Hence, the true winner of the election is candidate $a$.
We show that the attacker can make $p$ the winner if and only if $\langle G, \ell\rangle$ is a
yes-instance of {\sc Independent Set}, i.e., there is an independent set of size $\ell$ in $G$.

Suppose first that there is an independent set $V'\subseteq V$, $|V'| = \ell$, in $G$.
The following manipulation strategy makes $p$ the winner.
For every $u \in V'$, change the winner of district $D_u$ from $a$ to $j_u$,
and for every $e\in E$ such that $u\in e$, change the winner of district $D_{e,u}$ from $j_u$ to $j_e$.
Note that since $V'$ is an independent set, the weight of each candidate $j_e$, $e\in E$,
increases by at most $5$.
Let $\omega_c$ denote the weight of each candidate $c \in C$ after manipulation.
We have
$\omega_a = 5\mu(\nu-\ell)$,
$\omega_p = 5\mu(\nu-\ell)+1$,
$\omega_{j_e} \in\{5\mu(\nu-\ell)-5, 5\mu(\nu-\ell)\}$ for each $e \in E$, and
$\omega_{j_u} = 5\mu(\nu-\ell )$ for each $u \in V$; thus, candidate $p$ becomes the winner of the election.

Conversely, suppose that the attacker has a manipulation that
makes $p$ the election winner;
for each $c\in C$, let $\omega_c$ be the weight of candidate $c$ after this manipulation.
Since $p$ cannot be made the winner in any additional district,
we have $\omega_p = 5\mu(\nu-\ell)+1$.
Let $V'$ be the set of all nodes $u\in V$ such that the attacker changes
the winner of $D_u$ from $a$ to $j_u$.
Since $\omega_a \le \omega_p$, we have $|V'|\ge \ell$; we will now argue
that $V'$ is an independent set. Indeed, consider a node $u\in V'$.
Changing the winner in $D_u$ from $a$ to $j_u$ increases the weight of $j_u$
by $5\mu$.
As we have $\omega_{j_u} \le \omega_p$, the manipulation
needs to reduce the weight of $j_u$ by $5|\{e\in E: u \in e\}|$.
The only way to do so is to change the winner from $j_u$ to $j_e$
in all districts $D_{e,u}$ with $u\in e$, thereby increasing the weight of $j_e$ by $5$.
Now, suppose that $x, y\in V'$ and $e=\{x, y\}\in E$. Then the manipulation
increases the weight of $j_e$ by $10$,
so we have $\omega_{j_e}= 5\mu(\nu-\ell)+5>\omega_p$, a contradiction. Thus, $V'$ is an independent set.
\end{proof}

{\sc $\PD$-Man} remains NP-hard even if all districts have the same weight;
however, under this assumption this problem can be placed in NP, i.e.,
the unweighted variant of {\sc $\PD$-Man} is strictly easier than its weighted variant
unless NP$=\Sigma_2^P$ (which is believed to be highly unlikely).

\begin{theorem}
\label{thm:PD-Man-NPhardness2}
{\sc $\PD$-Man} is {\em NP}-complete when $w_i=1$ for all $i \in [k]$.
\end{theorem}
\begin{proof}
To see that {\sc $\PD$-Man} is in NP when $w_i=1$ for all $i\in [k]$, it suffices to note that
{\sc $\PD$-Rec} is in P under this assumption (Theorem~\ref{thm:PD-Rec-unweighted}).
To prove that {\sc $\PD$-Man} remains NP-hard even in this case,
we again provide a reduction from {\sc Independent Set}; see Definition~\ref{def:IS}.

Given an instance $\langle G, \ell\rangle$ of {\sc Independent Set},
where $G=(V, E)$,
we construct an instance of {\sc $\PD$-Man} as follows.
Let $\nu=|V|$, $\mu=|E|$, and for each $u\in V$ let $\deg(u)$ denote the degree of vertex $u$ in $G$;
without loss of generality, we can assume that $\mu>0$ and $\deg(u)>0$ for all $u\in V$.
Let $A_V = \{a_u: u \in V\}$,  $A'_V=\{b_u : u \in V \}$, $A_E=\{a_e : e \in E \}$,
and set $C = A_V\cup A'_V\cup A_E \cup \{p\}$,
where $p$ is the attacker's preferred candidate; thus, $|C| = 2\nu + \mu + 1$.
The tie-breaking order $\succ$ is defined so that $p \succ c$ for all $c \in C \setminus \{p\}$,
and $c\succ c'$ for all $c\in A_V$, $c'\in A_E$.
We create the following districts (note that the weight of each district is $1$).
\begin{itemize}
\item
    For every edge $e=\{x,y\} \in E$, we create two districts $D_{e,x}$ and $D_{e,y}$
    with $5$ voters each.
    In each such district $D_{e,u}$ there are two voters who vote for $a_u$ and
    three voters who vote for $a_e$.
    We set $\gamma_{e, u}=1$; thus, the attacker can change the winner in this district
    from $a_e$ to $a_u$.
\item
    For every vertex $u\in V$, we create a district $D_u$ with two voters who vote
    for $a_u$ and three voters who vote for $b_u$.
    We set $\gamma_u=1$; thus, the attacker can change the winner in this district
    from $b_u$ to $a_u$.
\item
    There are also some districts that cannot be manipulated (i.e., $\gamma=0$);
    for concreteness, we assume that each such district has five voters, and they
    all vote for the same candidate:
\begin{itemize}
\item For each $e \in E$, there are $\mu-1$ districts where the winner is $a_e$;
\item For each $u \in V$, there are $\mu-\deg(u)$ districts
     where the winner is $a_u$;
\item There are $\mu$ districts where the winner is $p$.
\end{itemize}
\end{itemize}
The budgets are $B_\calA = 2\mu + \ell$ and $B_\calD = \ell$.
Thus, we have $\SW(p) = \mu$, $\SW(a_e) = \mu+1$ for each $e\in E$,
$\SW(a_u) = \mu-\deg(u) < \mu$ and $\SW(b_u) = 1$ for each $u \in V$.
Consequently, the true winner is one of the candidates in $A_E$.

We will now argue that $G$ admits an independent set of size $\ell$ if and only if
there is a winning strategy for the attacker.

Suppose first that $V' \subseteq V$ is an independent set of size $\ell$.
Consider the following strategy for the attacker,
which changes votes in exactly $B_\calA$ districts:
\begin{itemize}
\item For each $e=\{x,y\} \in E$, change the winner of $D_{e,x}$ from $a_e$ to $a_x$,
    and the winner of $D_{e,y}$ from $a_e$ to $a_y$.
\item For each $u \in V'$, change the winner of $D_u$ from $b_u$ to $a_u$.
\end{itemize}
Let $\omega_c$ denote the weight of each candidate $c \in C$ after this manipulation.
We have $\omega_p = \mu$, $\omega_{a_e} = \mu-1$ for each $e \in E$,
$\omega_{a_u} = \mu$ for each $u \in  V \setminus V'$,
$\omega_{a_u} = \mu+1$ for each $u \in V'$, and $\omega_{b_u} = 0$ each $u \in V$.
Hence, in the manipulated instance the winner is chosen from
$\{a_u: u\in V'\}$ according to the tie-breaking rule.

Even though $p$ does not win the election at this point, we will now show that
$p$ becomes the winner once the defender respond optimally to this manipulation.

First, we show that the defender can make $p$ win.
To this end, for each $u\in V'$ the defender can pick one edge $e^u$
such that $u\in e^u$ and demand a recount in district $D_{e^u, u}$; altogether,
this strategy requires recounting $\ell=B_\calD$ districts.
Since $V'$ is an independent set,
after the recount the weight of each candidate $a_e$, $e\in E$, is at most $\mu$,
and also the weight of each candidate $a_u$, $u\in V$, is at most $\mu$.
Since $\omega_p=\mu$ and $p$ is favored by the tie-breaking rule, $p$ becomes
the election winner.

We will now argue that for every candidate $a$ that can be made the election winner by recounting
at most $\ell$ districts we have $\SW(a)\le \SW(p)$; since defender breaks ties
according to $\succ$, this proves that the defender will choose a recounting
strategy that makes $p$ win.
To see this, suppose for the sake of contradiction that there is a recounting strategy
that results in a candidate $a$ with $\SW(a)>\SW(p)$ becoming the election winner.
Note that $\SW(a)>\SW(p)$ implies that $a\in A_E$ and hence $\omega_a=\mu-1$.
Let $\omega'_c$ denote the weight of each candidate $c \in C$ after the recount.
The attacker does not transfer any district to $a$, which implies that
$\omega'_a\le \SW(a)=\mu+1$.
On the other hand, since $\omega'_p=\mu$, and the tie-breaking rule
favors $p$ over all other candidates,
we have $\omega'_a\ge \mu+1$. Thus, $\omega'_a=\mu+1$. This means that
$\omega'_a-\omega_a=2$, i.e., if $a=a_e$ and $e=\{x, y\}$, both $D_{e, x}$ and $D_{e, y}$
are recounted. We will now argue that $x, y\in V'$.
Indeed, for each $u\in V'$ we have $\omega_{a_u}=\mu+1$;
on the other hand, $a_u \succ a$ and hence $\omega'_{a_u} < \omega'_{a} = \mu + 1$.
Thus, the defender must demand that for each $u\in V'$ the district $D_u$
is recounted; since $B_\calD=\ell$, the set of recounted districts is exactly $V'$,
and hence $x, y\in V'$, as claimed. But this is a contradiction, since $\{x, y\}\in E$,
and $V'$ is an independent set. This proves that if $\langle G, \ell\rangle$
is a yes-instance of {\sc Independent Set}, there is a winning strategy for the attacker.

Conversely, suppose that $G$ has no independent set of size $\ell$.
Consider an attack that changes votes in at most $B_\calA$ districts.
For each $c\in C$, let $\omega_c$ denote the weight of candidate $c$ after the attack.
Note that $\omega_p=\mu$; moreover,
any attack can only increase the weight of candidates
in $A_V$, and the weight of any such candidate after the attack is at most $\mu+1$.
Let $V' = \{u \in V : \omega_{a_u} = \mu+1\}$ and $C' = \{a_u \in C : u \in V'\}$.
We consider three cases:
\begin{itemize}
\item
$|V'| > B_\calD$. Since recounting a district only reduces the weight of one candidate,
the weight of some candidate $a_u\in C'$ will still be $\mu+1$ after the recount,
so $p$ will be beaten by $a_u$.
\item
$|V'| \le B_\calD$, $V'$ is not an independent set.
Pick an edge $e^*=\{x, y\}$ such that $x, y\in V'$, and consider the following recounting strategy.
For each $u\in V'\setminus\{x, y\}$, the defender picks one edge $e^u$
such that $u\in e^u$, and demands a recount in districts $D_{e^u, u}$
for each $u\in V'\setminus\{x, y\}$ as well as in $D_{e^*, x}$ and in $D_{e^*, y}$.
This recounting strategy requires recounting $|V'|\le B_\calD$ districts,
reduces the weight of every candidate $c \in C'$ by $1$ and increases the
weight of $a_{e^*}$ by $2$. Thus, after the recount the weight of $a_{e^*}$
is $\mu+1$, whereas the weights of all candidates in $C\setminus A_E$
do not exceed $\mu$, so the winner is a candidate $a\in A_E$.
Since $\SW(a)>\SW(p)$, this means that $p$ cannot win after the recount.
\item
$|V'|\le B_\calD$, $V'$ is an independent set.
Then by our assumption $|V'|<\ell=B_\calD$.
Consider an edge $e^*=\{x, y\}$ with $x\in V'$, $y\not\in V'$.
For each $u\in V'\setminus\{x\}$, the defender can pick one edge $e^u$
such that $u\in e^u$, and demand a recount in districts $D_{e^u, u}$
for each $u\in V'\setminus\{x\}$ as well as in $D_{e^*, x}$ and in $D_{e^*, y}$.
This strategy requires recounting $|V'|+1\le B_\calD$ districts and ensures
that after the recount the weight of $e^*$
is $\mu+1$, whereas the weights of all candidates in $C\setminus A_E$
are at most $\mu$, so the winner is a candidate $a\in A_E$.
Since $\SW(a)>\SW(p)$, this means that $p$ cannot win after the recount.
\end{itemize}
Hence, the attacker cannot win in any case. This completes the proof.
\end{proof}

Theorem~\ref{thm:PD-Man-NPhardness1}
holds even for $B_\calD=0$, but for
Theorems~\ref{thm:PD-Man-hardness} and~\ref{thm:PD-Man-NPhardness2}
this is not the case. Indeed,
{\sc $\PD$-Man} is in NP when $B_\calD=0$, since the attacker simply needs
to guess a manipulation and check whether it makes $p$ the winner. The unweighted problem
(Theorem~\ref{thm:PD-Man-NPhardness2}) can be shown to be in P when $B_\calD=0$;
the argument uses a reduction to nonuniform bribery similar to
the one in the proof of Theorem~\ref{thm:PD-Rec-unweighted}.
Thus, recounting has a clear impact on the complexity of the attacker's problem.

\section{Regular Manipulations}\label{sec:regular}
In our model, the attacker does not have to transfer votes to his preferred candidate $p$
in the manipulated districts; indeed, he may even choose to transfer votes {\em from} $p$
to another candidate. However, manipulations that give additional votes to candidates other than $p$
are counter-intuitive and may be difficult to implement in practice. Therefore, in this section
we study what happens if the attacker is limited to transferring votes (in case of $\PV$)
or vote weight (in case of $\PD$) to his preferred candidate $p$.

\begin{definition}[Regular manipulation]
Let $p$ be the preferred candidate of the attacker.
A manipulation $(M, \tildevv)$
is said to be {\em regular}
if for every district $i\in M$ it holds that
\begin{itemize}
\item
the voting rule is $\PV$ and
$
\tilde{v}_{ia} \le v_{ia} \text{ for all $a\in C\setminus\{p\}$};
$
\item
the voting rule is $\PD$
and in $\tildevv$ candidate $p$ is the winner in each district in $M$.
\end{itemize}
\end{definition}

The difference between our general model and the one where the attacker is limited to using regular
manipulations is similar to the difference between swap bribery and shift bribery \citep{EFA09}:
in swap bribery the attacker can change the vote in any way he likes subject to budget constraints,
while in shift bribery he is limited to shifting his preferred candidate in voters' rankings.

One may expect that the restriction to regular manipulations is without loss of generality:
indeed, why would the attacker want to transfer votes to candidates other than $p$? However,
our next example shows that this intuition is incorrect.

\begin{example}\label{ex:regular}
We show an example for $\PV$; the example also works for $\PD$ by setting $w_i=n_i$ for every $i\in[k]$. Consider an instance with $3$ candidates $\{a,b,p\}$ and $19$
voters who are distributed to $12$ districts.
The vote profile is as follows:
\begin{center}
\small
\begin{tabular}{c c c c c}
\noalign{\hrule height 1pt}
Candidate & $D_1$ & $D_2$ & \quad $D_3,\dots,D_8$ \quad & \quad $D_9,\dots, D_{12}$ \\
\noalign{\hrule height 1pt}
$a$  & $0$ & $3$ & $1$ & $0$ \\
$p$  & $6$ & $0$ & $0$ & $0$ \\
$b$  & $0$ & $0$ & $0$ & $1$ \\
\noalign{\hrule height 1pt}
\end{tabular}
\end{center}

\noindent
Also, $B_\calA=2$, $B_\calD=1$, and $\gamma_i=n_i$ for all $i\in\{1, \dots, 12\}$.
The true winner is candidate $a$ with $9$ votes, compared to the $6$ votes of $p$ and the $4$ votes of $b$.
No regular manipulation can make $p$ win: no matter what the attacker does,
by recounting at most one district the defender can ensure that $a$ gets at least $8$
votes and $p$ gets at most $7$ votes.

Now, consider a non-regular manipulation that distorts all votes in $D_1$ 
in favor of $b$, and all votes in $D_2$ 
in favor of $p$. Then in the distorted profile
$a$ has $6$ votes and $p$ has $3$ votes, and $b$ wins with $10$ votes.
If the defender does not recount $D_1$, 
$b$ remains the winner after recounting,
and if she does recount it, $p$ becomes the winner. Crucially, since $\SW(b)<\SW(p)$,
the defender prefers the latter option, so $p$ wins after the recount.
\hfill\qed
\end{example}

Example~\ref{ex:regular} shows that only considering regular manipulations
may be suboptimal for the attacker. However, the attacker may be limited to
regular manipulations by practical considerations. For instance, the election officials
in the manipulated districts may find it difficult to follow complex instructions.
Thus, it is interesting to understand
if focusing on regular manipulations affects the complexity of the problems we consider.

The following observation will be useful for our analysis.
\begin{proposition}\label{prop:regular}
Let $\calR\in\{\PV, \PD\}$, and let $(M, \tildevv)$ be a winning regular manipulation.
Then for every recounting strategy $R\subseteq M$ 
it holds that after the recount $p$ is the election winner. 
\end{proposition}
\begin{proof}
Let $B=\{b\in C\setminus\{p\}:
\SW^\calR(b)<\SW^\calR(p)~\text{or}~\SW^\calR(b)=\SW^\calR(p), p\succ b\}$.
Since $M$ is a winning manipulation, 
the winner after recounting is either $p$ or some candidate in $B$;
we will show that, since $M$ is regular, the latter case is, in fact, impossible. 
For each $c\in C$, let $s_c$
denote the number of votes/vote weight of $c$ after the recount. 
Since $M$ is a regular manipulation, for each candidate $b\in B$
$$
s_b\le SW^\calR(b)\le \SW^\calR(p)\le s_p, 
$$
and if $b\succ p$, the second inequality is strict. Thus, $p$ beats every candidate
in $B$ after recounting, so no such candidate can be the election winner.
\end{proof}
By setting $R=\varnothing$ in Proposition~\ref{prop:regular}, we observe that
$p$ is the winner at $\tildevv$, i.e., the situaiton described in Example~\ref{ex:regular},
where $p$ does not win after the manipulation, but the defender is forced to make $p$
the election winner, cannot occur if the attacker is limited to regular manipulations.

In what follows, we consider the complexity of
{\sc $\calR$-Rec} and {\sc $\calR$-Man}
for $\calR\in \{\PV, \PD\}$ under the assumption that the attacker is limited to regular
manipulations; we denote these versions of our problems
by {\sc $\calR$-Rec-Reg} and {\sc $\calR$-Man-Reg}, respectively.
We first consider the defender's problem
({\sc $\calR$-Rec-Reg}) and then the attacker's problem ({\sc $\calR$-Man-Reg}).

\subsection{The Defender's Problem}

Let $\calR\in\{\PV, \PD\}$, and
consider a regular manipulation $(M, \tildevv)$. Recall that we assume that $|M| > B_\calD$.
Note that if $p$ is not a winner at $\tildevv$, the attacker necessarily
By Proposition~\ref{prop:regular}, we can assume
that $p$ is the winner at $\tildevv$. The defender can then try the following greedy strategy.
Initially, it defines the set of {\em provisional winners} to consist of $p$.
Then, for each $a\in C\setminus\{p\}$ such that $\SW^\calR(a)>\SW^\calR(p)$ or
$\SW^\calR(a)=\SW^\calR(p)$ and $a\succ p$
the algorithm sorts the districts in $M$ in non-increasing order in terms of the quantity
$
(v_{ia} - v_{ip}) - (\tilde{v}_{ia} - \tilde{v}_{ip})
$
for $\PV$, and the quantity
$
(w_{ia} - w_{ip}) - (\tilde{w}_{ia} - \tilde{w}_{ip})
$
for $\PD$; ties are broken arbitrarily.
Next, it checks what happens if the first $B_\calD$ districts in this order
are recounted; if this results in a candidate $b\in C\setminus\{p\}$ with $\SW^\calR(b)>\SW^\calR(p)$
or $\SW^\calR(b)=\SW^\calR(p)$, $b \succ p$, winning the election,
the defender adds $b$ to the set of {\em provisional winners}.
Finally, it outputs the provisional winner with the maximum social welfare,
breaking ties according to $\succ$.
We refer to this algorithm as {\em greedy recounting};
note that its running time is polynomial in the input size.

\begin{lemma}
\label{lem:Rec-regular-better}
Let $\calR\in\{\PV, \PD\}$. Suppose that the attacker uses a regular manipulation $(M, \tildevv)$.
Then greedy recounting outputs $p$ if and only if $(M, \tildevv)$ is a winning strategy for
the attacker.
\end{lemma}

\begin{proof}
We provide the proof for $\calR = \PV$; the proof for $\calR = \PD$ is obtained
by replacing candidates' vote counts with weights.
Consider the attacker's strategy $(M, \tildevv)$.
Given a set of districts $R\subseteq M$ and a candidate $a\in C$,
let $s_a(R)$ denote the number of votes that candidate $a$ gets after
the attacker manipulates according to $(M, \tildevv)$
and the defender recounts the districts in $R$.
Given two candidates $a, b\in C$ and a subset of districts $R\subset M$,
we write $s_a(R) \rhd s_b(R)$ if $s_a(R)>s_b(R)$ or $s_a(R)=s_b(R)$ and $a\succ b$.

Since $(M, \tildevv)$ is a regular manipulation, for every  $R\subseteq M$ we have
$$
s_p(R) \ge \SW(p)\quad\text{and}\quad
s_c(R) \le \SW(c) \text{ for all $c\in C\setminus\{p\}$}.
$$

Suppose that $(M, \tildevv)$ is not a winning strategy for the attacker.
Then there exists a subset $R^*$ of at most $B_\calD$ districts
such that after recounting $R^*$ the winner is a candidate
$a\in C\setminus\{p\}$ such that $\SW(a) > \SW(p)$ or $\SW(a)=\SW(p)$, $a\succ p$.
We can assume without loss of generality that $|R^*|=B_\calD$.
Indeed, suppose that $|R^*|<B_\calD$ and $a$ wins after recounting the districts in $R^*$.
Let $Q$ be an arbitrary set of $B_\calD$ districts such that $R^*\subset Q\subseteq M$,
and suppose that once the votes in $Q$ are recounted, the winner is $c$.
Since $(M, \tildevv)$ is a regular manipulation, we have
\begin{align*}
\SW(c)\ge s_c(Q)\ge s_a(Q)\ge s_a(R^*)\ge s_p(R^*)
      \ge \SW(p).
\end{align*}
If any of these inequalities is strict, we have $\SW(c)>\SW(p)$.
Otherwise, the second inequality implies $c\succ a$ and the fourth inequality implies
$a\succ p$, so we have $\SW(c)=\SW(p)$, $c\succ p$. In either case,
recounting the districts in $Q$ results in an outcome that the defender prefers to $p$.

We will argue that the greedy recounting algorithm does not output $p$.
Let $R_a$ be the set of districts recounted by this algorithm
when it considers candidate $a$ (i.e., $R_a$ contains the first $B_\calD$ districts
in non-increasing order of the quantity $(v_{ia}-v_{ip}) - (\tilde{v}_{ia} - \tilde{v}_{ip})$),
and let $b$ be the winner after recounting the districts in $R_a$.
Consider the following possibilities.
\begin{itemize}
\item
$b\neq p$.
If $\SW(b)>\SW(p)$ or $\SW(b)=\SW(p)$, $b\succ p$,
the algorithm adds $b$
to the set of provisional winners and thus does not output $p$.
Otherwise, we have
$$
s_p(R_a) \ge \SW(p) \ge \SW(b) \ge s_b(R_a);
$$
if $b\succ p$, the second inequality is strict.
Consequently, $s_p(R_a)\rhd s_b(R_a)$, so
$b$ cannot be the winner after recounting the districts in $R_a$, a contradiction.

\item
$b=p$.
By our choice of $R_a$ and the fact that $|R_a|=|R^*|$ we have
\begin{align*}
0 &\le s_a(R^*) -s_p(R^*)\\
  &= \sum_{i\not\in R^*}(\tilde{v}_{ia}-\tilde{v}_{ip})+\sum_{i\in R^*} (v_{ia}-v_{ip})\\
  &= \sum_{i\in [k]}(\tilde{v}_{ia}-\tilde{v}_{ip})+
    \sum_{i\in R^*} (v_{ia}-v_{ip}-\tilde{v}_{ia}+\tilde{v}_{ip})\\
  &\le \sum_{i\in [k]}(\tilde{v}_{ia}-\tilde{v}_{ip})+
    \sum_{i\in R_a} (v_{ia}-v_{ip}-\tilde{v}_{ia}+\tilde{v}_{ip})\\
  &= \sum_{i\not\in R_a}(\tilde{v}_{ia}-\tilde{v}_{ip})+\sum_{i\in R_a} (v_{ia}-v_{ip})\\
  &= s_a(R_a) -s_p(R_a).
\end{align*}
Combining this with
the fact that $s_p(R_a) \rhd s_a(R_a)$, we conclude that $s_p(R_a) = s_a(R_a)$
and $p\succ a$.
Thus, all inequalities above are, in fact, equalities,
so in particular $s_p(R^*) = s_a(R^*)$. Together with $p\succ a$
this implies that $s_p(R^*)\rhd r_a(R^*)$,
a contradiction with our assumption that $a$ is the winner after recounting $R^*$.
\end{itemize}
This completes the proof.
\end{proof}

Notably, greedy recounting does not constitute an algorithm for {\sc $\calR$-Rec-Reg}:
it is unable to decide whether there is a recounting strategy that results
in a specific candidate becoming the election winner.
However, it serves as a $1/2$-approximation algorithm
for the defender: it outputs a candidate $a$ such that for every candidate $a'$ that
can be made a winner by recounting at most $B_\calD$ districts it holds that $\SW(a)\ge \SW(a')/2$.

\begin{theorem}
\label{thm:Rec-approx-regular}
Greedy recounting is a $1/2$-approximation algorithm
for the optimization versions of {\sc $\PV$-Rec-Reg} and {\sc $\PD$-Rec-Reg}.
\end{theorem}

\begin{proof}
We focus on PV; the analysis can be adapted for PD, by modifying the notation
so as to take into account the weights of the candidates rather than their vote counts.

Consider an instance with a set of candidates $C$, $|C|=m$,
and let $p$ be the attacker's preferred candidate.
Suppose that the attacker uses a regular manipulation $(M,\tildevv)$;
we will assume that $p$ is the winner in the manipulated instance,
as otherwise the attacker does not have an incentive to manipulate.
For each $c\in C$, let $s_c$ denote the vote count of candidate $c$ in the manipulated instance.
If $p$ is the winner before the manipulation or if no recounting
strategy can change the outcome, then greedy recounting is
trivially optimal. Hence, in the remainder of the proof we assume that
there is a candidate $b \neq p$ such that $\SW(b)> \SW(p)$ or $\SW(b)=\SW(p)$, $b\succ p$,
such that the defender can make $b$ win; let $c$ be the defender's most
preferred candidate with this property, and let $R$ be a recounting strategy
that results in $c$ becoming the winner after a recount.
We consider the round in which greedy recounting examines candidate $c$;
suppose that greedy recounting selects a subset of districts $G$.
Let $A = C \setminus \{c,p\}$ denote the set of the remaining $m-2$ candidates.

We define the following pairwise disjoint sets of districts:
\begin{itemize}
\item $I_G=G\setminus R$;
\item $I_O=R\setminus G$;
\item $I_{OG}=R\cap G$;
\item $I_{\overline{OG}}=M\setminus(R\cup G)$.
\end{itemize}
Given a set of districts $I\subseteq M$ and
a subset of candidates $J\subseteq C\setminus\{p\}$,
let
$$
\Delta(I, J)= \sum_{i \in I}\sum_{a \in J}(v_{ia}-\tilde{v}_{ia})
$$
denote the total number of votes in $I$ that are transferred by the attacker from candidates in $J$
to $p$;
if $J$ or $I$ is a singleton, we omit the curly braces and write $\Delta(I, j)$
or $\Delta(i, J)$, respectively.
Since $(M, \tildevv)$ is a regular manipulation,
we have
\begin{align}
s_p &= \SW(p) + \Delta(M, c) + \Delta(M, A),\label{eq:score-p}\\
s_c &= \SW(c) - \Delta(M, c),\label{eq:score-c}\\
s_a &= \SW(a) - \Delta(M, a)~\text{for each $a\in A$}\label{eq:score-a}.
\end{align}
Since recounting the districts in $R = I_O \cup I_{OG}$
ensures that $c$ becomes the winner, we obtain
\begin{align}\label{eq:cp1}
s_c + \Delta(I_O \cup I_{OG}, c)
\geq s_p - \Delta(I_O \cup I_{OG}, c) - \Delta(I_O \cup I_{OG}, A);
\end{align}
if $p\succ c$, this inequality is strict.

Next, let us focus on the behavior of the greedy recounting. Let $g \in I_G$ and $o \in I_O$.
Since the greedy algorithm selects $g$, but not $o$, we have
\begin{align*}
(v_{gc} - v_{gp}) - (\tilde{v}_{gc} - \tilde{v}_{gp})
\ge (v_{oc} - v_{op}) - (\tilde{v}_{oc} - \tilde{v}_{op}).
\end{align*}
Since
$v_{ic} - \tilde{v}_{ic} = \Delta(i, c)$ and
$\tilde{v}_{ip} - v_{ip} = \Delta(i, c) + \Delta(i, A)$ for every $i \in M$,
we then obtain
\begin{align*}
2 \Delta(g, c) + \Delta(g, A) \ge
2 \Delta(o, c) + \Delta(o, A).
\end{align*}
Since $|G|=B_\calD$, $|R|\le B_\calD$,
we have $|I_G|\ge |I_O|$. Pick a subset of districts $I'_G\subseteq I_G$
with $|I'_G|=|I_O|$.
We can pair each $o \in I_O$ with a unique $g \in I'_G$,
and add all corresponding inequalities to get
$$
 2\Delta(I'_G, c) + \Delta(I'_G, A)
 \ge 2\Delta(I_O, c) + \Delta(I_O, A);
$$
since $\Delta(i, b)\ge 0$ for all $b\in A\cup\{c\}$ and all $i\in I_G\setminus I'_G$,
we get
\begin{align}\label{eq:greedy}
 2\Delta(I_G, c) + \Delta(I_G, A)
 \ge 2\Delta(I_O, c) + \Delta(I_O, A).
\end{align}

By adding inequalities \eqref{eq:cp1} and \eqref{eq:greedy}, we obtain
\begin{align*}
&  s_c  +  \Delta(I_O \cup I_{OG}, c)+ 2\Delta(I_G, c) + \Delta(I_G, A)\\
&\ge s_p - \Delta(I_O \cup I_{OG}, c) - \Delta(I_O \cup I_{OG}, A)+ 2\Delta(I_O, c) + \Delta(I_O, A),
\end{align*}
or, simplifying,
\begin{align}\label{eq:cp2}
s_c + \Delta(I_G \cup I_{OG}, c)
\ge s_p - \Delta(I_G \cup I_{OG}, c) - \Delta(I_G \cup I_{OG}, A);
\end{align}
if $p\succ c$, then Inequality~\eqref{eq:cp1} is strict and hence Inequality~\eqref{eq:cp2}
is strict as well.
Inequality~\eqref{eq:cp2} means that after recounting the districts in
$G = I_G \cup I_{OG}$, $c$ beats $p$, i.e., the winner is either $c$
or another candidate $a\in A$.
To conclude the proof, it suffices to show that if the winner in the recounted instance
is some candidate $a\in A$ then $\SW(a) \geq \frac{1}{2}\SW(c)$.

By substituting expressions for $s_c$ and $s_p$ from~\eqref{eq:score-c} and~\eqref{eq:score-p},
we can write Inequality~\eqref{eq:cp2} as
\begin{align*}
\SW(c) - \Delta(I_O \cup I_{\overline{OG}}, c)
\geq \SW(p) + \Delta(I_O \cup I_{\overline{OG}}, c) +
\Delta(I_O \cup I_{\overline{OG}}, A).
\end{align*}
Since $\SW(p) \geq 0$ and $\Delta(I_O \cup I_{\overline{OG}}, A) \geq 0$,
it follows that
\begin{align}\label{eq:cc}
\Delta(I_O \cup I_{\overline{OG}}, c) \leq \frac{1}{2}\SW(c).
\end{align}
By our assumption, recounting the districts in $G = I_G \cup I_{OG}$ results in $a$
getting at least as many votes as $c$, so we obtain
\begin{align*}
s_a + \Delta(I_G \cup I_{OG}, a)
\geq s_c + \Delta(I_G \cup I_{OG}, c).
\end{align*}
By substituting expressions for $s_a$ and $s_c$ from~\eqref{eq:score-a} and~\eqref{eq:score-c},
we can rewrite this inequality as
\begin{align*}
\SW(a) - \Delta(I_O \cup I_{\overline{OG}}, a)
\geq \SW(c) - \Delta(I_O \cup I_{\overline{OG}}, c).
\end{align*}
Finally, since $\Delta(I_O \cup I_{\overline{OG}}, a) \geq 0$,
from Inequality~\eqref{eq:cc} we obtain
\begin{align*}
\SW(a) \geq \frac{1}{2}\SW(c),
\end{align*}
as desired.
\end{proof}

In fact, the bound on the approximation ratio
provided by Theorem~\ref{thm:Rec-approx-regular} is essentially tight.

\begin{theorem}
\label{thm:Rec-inapprox-regular}
For any constant $\varepsilon>0$,
neither {\sc $\PV$-Rec-Reg} nor {\sc $\PD$-Rec-Reg} admits a polynomial-time
$(\frac{1}{2}+\varepsilon)$-approximation algorithm unless {\em P = NP}, even when $m=3$.
\end{theorem}

\begin{proof}
We focus on PV; the proof for PD follows by setting the weight of each district in the reduction below
to be equal to the number of voters therein.

We will show that if there is a $(\frac{1}{2}+\varepsilon)$-approximation algorithm for
{\sc $\PV$-Rec-Reg}, it can be used to solve {\sc Partition}; see Definition~\ref{def:Partition}.

Given an instance $X$ of {\sc Partition} with $|X|=\ell$,
we construct an instance of {\sc $\PV$-Rec-Reg} with a set of candidates $C=\{a, b, p\}$,
where $p$ is the attacker's preferred candidate,
as follows. Let $y = \sum_{x\in X} x$, and $z=\lceil y/\varepsilon\rceil$.
Without loss of generality, we assume that all integers in $X$ are divisible by $4$ and hence $y \ge 4$.
In what follows, we describe each district $D_i$ by a tuple $(v_{ia}, v_{ib}, v_{ip})$.
The districts are partitioned into the following three sets $I_1$, $I_2$ and $I_3$:
\begin{itemize}
\item
For each $x\in X$, there is a district in $I_1$ with votes
$(0, 2x\ell, 0)$, which are distorted to $(0, 0, 2x\ell)$.
\item
$I_2$ consists  of $2z\ell$ districts with votes $(1,0,0)$,
which are distorted to $(0,0,1)$.
\item
$I_3$ consists of two districts with votes
$(2z\ell + y\ell + 2\ell, 0, 0)$ and $(0, 2z\ell,0)$,
which are not distorted.
\end{itemize}
Finally, the budget of the defender is $B_\calD = \ell-1$.

Since votes are transferred to $p$ only, the manipulation is regular.
The vote counts of the candidates before and after the manipulation are as follows:
\begin{center}
\small
\begin{tabular}{c l l}
\noalign{\hrule height 1pt}
	& True vote counts ($\SW$) 	& Distorted vote counts \\
\noalign{\hrule height 0.5pt}
$a$ 	& $4z\ell +  y \ell + 2\ell$ 	& $2z\ell + y\ell + 2\ell$ \\
$b$ 	& $2z\ell + 2y\ell$	 	& $2z\ell$  \\
$p$ 	& $0$			 	& $2z\ell + 2y\ell$ \\
\noalign{\hrule height 1pt}
\end{tabular}
\end{center}
Therefore, before the manipulation the winner is $a$,
and the manipulation makes $p$ the election winner.
Since
$$
\frac{\SW(c)}{\SW(a)} \le \frac{2z\ell+2y\ell}{4z\ell +  y \ell + 2\ell}
< \frac{1}{2} + \varepsilon\quad\mbox{for each}~c \in \{b, p\},
$$
any $(\frac{1}{2}+\varepsilon)$-approximation algorithm can decide whether
$a$ can be restored as the winner. We will now argue that this is equivalent
to deciding whether the given instance of {\sc Partition} is a yes-instance.

Suppose that $X$ is a yes-instance of {\sc Partition}, i.e., there exists a subset
$X' \subseteq X$ such that $\sum_{x \in X'}x = y/2$; note that $|X'|\le \ell-1$.
Then, by recounting the $|X'|$ districts of $I_1$ that
correspond to the integers in $X'$, the defender lowers the vote count of $p$ by $y\ell$ and increases
the vote count of $b$ by $y\ell$. As a result,
$a$ gets $2z\ell + y\ell + 2\ell$ votes,
$b$ gets $2z\ell + y\ell$
votes, and $p$ gets $2z\ell + y\ell$ votes.
Therefore, $a$ is restored as the election winner.

Conversely, suppose that there is no subset
$X' \subseteq X$ such that $\sum_{x \in X'} x = y/2$.
Since all integers in $X$ are divisible by $4$,
$y/2$ is even and hence for any $X' \subseteq X$ we have
$|\sum_{x \in X'} x - y/2 | \geq 2$.
Suppose that the defender recounts districts in $I_1$
that correspond to a subset $X'\subseteq X$ as well as $q$ districts in $I_2$;
let $u = \sum_{x\in X'}x$.
Since $q< \ell$, the vote count of candidate $a$ after the recount is
$$
2z\ell + y\ell + 2\ell + q < 2z\ell + y\ell + 3\ell.
$$
If $u \geq y/2+2$, then the vote count of $b$ after the recount is
$$
2z\ell + 2 u\ell \geq 2z\ell + y\ell + 4\ell.
$$
Otherwise, $u \leq y/2-2$.
Since $q<\ell$, the vote count of $p$ after the recount is
\begin{align*}
2z\ell - q + 2\ell\sum_{x \in X \setminus X'}x
&= 2z\ell - q + 2y\ell - 2 u\ell\\
&> 2z\ell + 2y\ell - 2 u\ell - \ell \\
&\ge 2z\ell + y\ell + 3\ell.
\end{align*}
Therefore, in either case one of $b$ or $p$ gets more votes than $a$, and the theorem follows.
\end{proof}

\subsection{The Attacker's Problem}
Greedy recounting also plays an important role in our analysis of {\sc $\calR$-Man-Reg}.
Indeed, even though greedy recounting does not constitute an algorithm for
{\sc $\calR$-Rec-Reg}, Lemma~\ref{lem:Rec-regular-better}
suffices to establish that {\sc $\calR$-Man-Reg} is in NP:
the attacker can guess a regular manipulation and
use greedy recounting to verify whether it is successful.
For PV, this complexity upper bound is tight:
one can check that in the hardness proofs in Theorem~\ref{thm:PV-Man-hardness}
the attacker's successful manipulation strategy is regular,
and hence {\sc $\PV$-Man-Reg} is NP-complete.
We summarize these observations in the following theorem.

\begin{theorem}\label{thm:Man-hardness-regular}
{\sc $\PV$-Man-Reg} is {\em NP}-complete.
The hardness result holds even if $m=3$ or if the input vote profile and district
weights are given in unary.
\end{theorem}

We cannot use the same approach to show that {\sc $\PD$-Man-Reg} is NP-hard:
the hardness proofs in Theorems~\ref{thm:PD-Man-hardness}--\ref{thm:PD-Man-NPhardness2} 
rely on the attacker using a non-regular strategy.
In fact, it turns out that {\sc $\PD$-Man-Reg} is polynomial-time solvable, 
i.e., for PD focusing on regular manipulations brings down
the complexity of the attacker's problem from $\Sigma_2^P$ to P.

\begin{theorem}\label{thm:PD-Man-poly-regular}
{\sc $\PD$-Man-Reg} can be solved in polynomial time.
\end{theorem}

\begin{proof}
Let $p$ be the attacker's preferred candidate, and 
let $A = \{c \in C: \SW(c) > \SW(p) \text{ or } \SW(c) = \SW(p), c \succ p \}$ 
be the set of candidates that are preferred to $p$ by the defender.
For each $c \in C\setminus\{p\}$, we denote by $S_c$ the set of districts 
that have $c$ as their true winner and can be  
manipulated in favor of $p$. Let $S = \bigcup_{c \in C\setminus\{p\}} S_c$
denote the set of all districts that can be manipulated in favor of $p$.
Note that for every $c\in C$ the set $S_c$ can be computed 
efficiently: the problem of deciding if the winner of district $D_i$ 
can be changed to $p$ can be viewed as an instance of nonuniform bribery under Plurality with 
prices in $\{0, 1\}$ and budget $\gamma_i$, and nonuniform bribery is in P
for the Plurality rule
(see the proof of Theorem~\ref{thm:PD-Rec-unweighted} 
for the definition of nonuniform bribery and references).

Since the manipulation is regular, the attacker's strategy can be identified 
with a subset $M \subseteq S$.
Let $\ell = \min\{B_\calA, |S|\}$ be the maximum number of districts that can be manipulated.
For any set $Q \subseteq S$, $|Q| \le \ell$, let $f(Q)$ be the set that 
consists of $\ell - |Q|$ heaviest districts in $S \setminus Q$, 
with ties broken arbitrarily; thus, $|f(Q) \cup Q| = \ell$.
Our algorithm is based on the following lemma.

\begin{lemma}\label{lem:pd-reg}
Consider a subset $Q \subset S$ such that
there exists a winning regular manipulation $M$, $|M| \le \ell$, with $Q\subset M$.
Suppose that when the attacker manipulates the districts in 
$Q \cup f(Q)$, there is a candidate $a\in A$ such that the defender 
can make $a$ beat $p$ by recounting at most $B_\calD$ districts. 
Let $S_a^{\max}=\arg\max_{j \in S_a \setminus Q} w_j$. 
Then 
\begin{itemize}
\item[(i)] $S_a\setminus Q\neq\varnothing$, and
\item[(ii)] for each $i\in S_a^{\max}$ there is a winning regular manipulation $M'$, 
$|M'| \le \ell$, with $Q \cup \{i\} \subseteq M'$.
\end{itemize}
\end{lemma}
Before we prove this lemma, we will explain how to use it to find 
a winning regular manipulation if it exists. The algorithm proceeds as follows.
\begin{itemize}
\item[1.]
Set $Q = \varnothing$.
\item[2.]
Apply greedy recounting to $Q\cup f(Q)$ to check whether $Q\cup f(Q)$
is a winning regular manipulation. If yes, terminate and return $Q\cup f(Q)$.
Otherwise greedy recounting returns a candidate $a\in A$ such that 
the defender can make $a$ beat $p$ by recounting at most $B_\calD$ districts.
\item[3.]
If $S_a \setminus Q = \varnothing$ or $|Q|=\ell$, then output $\varnothing$.
Otherwise, select an arbitrary $i \in S_a^{\max}$, 
set $Q \leftarrow Q \cup \{i\}$, and go back to Step 2.
\end{itemize}
By Lemma~\ref{lem:Rec-regular-better}, if the algorithm returns $Q\cup f(Q)$
at the end of Step 2, then $Q\cup f(Q)$ is a winning regular manipulation.
Otherwise, by Lemma~\ref{lem:pd-reg}, there is no winning strategy.
This shows that our algorithm is correct. To see that it runs in polynomial time, 
note that every execution of Step 2 increases $|Q|$ by $1$,  
and $|Q|$ is bounded from above by $\ell$.

To complete the proof, it remains to prove Lemma~\ref{lem:pd-reg}
\begin{proof}[Proof of Lemma~\ref{lem:pd-reg}]
Suppose 
that there exist $Q$, $M$ and $a$  
that satisfy the conditions in the statement of the lemma. 
For each candidate $c\in C$ and each $X\subseteq S$, 
let $s_c(X)$ denote the weight of $c$ 
after the districts in $X$ have been manipulated in favor of $p$.
We prove each claim of the lemma separately.

\paragraph{Proof of claim~(i).}
We will prove a stronger claim, 
namely, that $M \cap (S_a \setminus Q) \neq \varnothing$.

Suppose for the sake of contradiction that $M \cap (S_a \setminus Q) = \varnothing$.
We will argue that in this case $Q\cup f(Q)$ is a winning manipulation, 
thereby obtaining a contradiction with the assumptions of the lemma.
To this end, we consider an arbitrary recounting strategy $R'\subseteq Q\cup f(Q)$, 
$|R'| \le B_\calD$,
transform it into a recounting strategy $R\subseteq M$, and use the fact 
that $M$ is a winning manipulation.

Since $Q \subseteq M$, we have $M \cap S_a \subseteq Q$ and 
$(M\setminus Q) \cap S_a = \varnothing$. Hence,
\begin{align}\label{eq:PD-Man-poly-regular:1}
\sum_{i \in M\setminus Q} w_{ia} = \sum_{i \in (M\setminus Q) \cap S_a} w_{i} = 0.
\end{align}
Fix a recounting strategy $R' \subseteq Q\cup f(Q)$, $|R'| \le B_\calD$. 
Let $R''$ be the set of  $\min \left\{ |R' \cap f(Q)|, |M\setminus Q| \right\}$ 
heaviest districts in $M\setminus Q$, 
and set $R = (R' \cap Q ) \cup R''$.
Note that $R\subseteq M$ and $|R|\le B_\calD$:
we have 
$|R| \le |R' \cap Q| + |R' \cap f(Q)| \le |R'| \le B_\calD$.
Moreover, $Q \setminus R = Q \setminus R'$, and 
\begin{subequations}\label{eq:PD-Man-poly-regular:2}
\begin{align*}
s_a \big( (Q\cup f(Q)) \setminus R' \big)
&= \SW(a) - \sum_{i \in (Q\cup f(Q)) \setminus R'} w_{ia} \\
&\le \SW(a) - \sum_{i \in Q \setminus R'} w_{ia} \\
&= \SW(a) - \sum_{i \in Q \setminus R} w_{ia} - \sum_{i \in M\setminus Q} w_{ia} \\
&\le \SW(a) - \sum_{i \in Q \setminus R} w_{ia} - \sum_{i \in (M\setminus Q) \setminus R} w_{ia} \\
&= s_a ( M \setminus R ),
\tag{\ref{eq:PD-Man-poly-regular:2}}
\end{align*}
\end{subequations}
where the third transition follows by \eqref{eq:PD-Man-poly-regular:1}.

Next, we claim that
\begin{align*}
\sum_{i \in f(Q) \setminus R'} w_{i} \ge  \sum_{i \in (M \setminus Q) \setminus R} w_{i}.
\end{align*}
Indeed, if $|R''|=|M\setminus Q|$, then
$M\setminus Q\subseteq R$, so the right-hand side of this inequality is $0$,
and our claim is immediate. Otherwise, 
$|R''| = |f(Q)\cap R'| = |(M\setminus Q)\cap R|$
and $|f(Q)|=|M\setminus Q|$, i.e., 
both sums have the same number of summands.
Moreover, 
$f(Q)$ contains the heaviest $\ell-|Q|$ districts in $S \setminus Q$, 
$M\setminus Q\subseteq S\setminus Q$,
and $(M\setminus Q)\cap R$ consists of $|(M\setminus Q)\cap R|$
heaviest districts in $M\setminus Q$, so the claim follows.

We can now write
\begin{subequations}\label{eq:PD-Man-poly-regular:3}
\begin{align*}
s_p \big( (Q\cup f(Q)) \setminus R' \big)
&= \SW(p) + \sum_{i \in Q \setminus R'} w_{i} + \sum_{i \in f(Q) \setminus R'} w_{i} \\
&\ge \SW(p) + \sum_{i \in Q \setminus R} w_{i} + \sum_{i \in (M \setminus Q) \setminus R} w_{i} \\
&= s_p ( M \setminus R ). \tag{\ref{eq:PD-Man-poly-regular:3}}
\end{align*}
\end{subequations}
Combining inequalities~\eqref{eq:PD-Man-poly-regular:2} 
and \eqref{eq:PD-Man-poly-regular:3}, we obtain
\begin{align*}
s_a ( (Q\cup f(Q)) \setminus R') - s_p( (Q\cup f(Q))\setminus R')
\le s_a (M \setminus R) - s_p(M \setminus R).
\end{align*}
Since $M$ is a winning manipulation, we have $s_a(M\setminus R)\le s_p(M\setminus R)$, 
and if $a\succ p$, this inequality is strict. As this hold for any defender's strategy $R'$, 
it follows that $Q\cup f(Q)$ is a winning manipulation, too, which contradicts 
the assumptions of the lemma.

\paragraph{Proof of claim~(ii).}
We have established that $S_a\setminus Q\neq\varnothing$
and hence $S_a^{\max}\neq\varnothing$.
Now, suppose for the sake of contradiction that for some $i \in S_a^{\max}$ 
there is no winning regular manipulation $M'$ with 
$|M'| \le \ell$, $Q \cup \{i\} \subseteq M'$.
Since all districts in $S_a^{\max}$ 
are identical from both the attacker's and the defender's perspective, 
it holds that, in fact, for every $i \in S_a^{\max}$ 
there is no winning regular manipulation $M'$ with 
$|M'| \le \ell$, $Q \cup \{i\} \subseteq M'$.

We have argued that $M \cap (S_a \setminus Q) \neq \varnothing$;
pick some $j \in M \cap (S_a \setminus Q)$.
Since $Q\cup\{j\}\subseteq M$ and $M$ is a winning regular manipulation, 
it follows that $j \notin S_a^{\max}$. 
Pick some $i \in S_a^{\max}$ and set $M'=(M\setminus\{j\})\cup\{i\}$. 
We will now obtain a contradiction by showing 
that $M'$ is a winning regular manipulation.
Consider an arbitrary recounting strategy $R' \subseteq M'$, 
$|R'| \le B_\calD$.

\begin{itemize}
\item[(a)]
If $i \in R'$, let $R = (R'\setminus \{i\}) \cup \{j\}$ 
so that $|R| = |R'| \le B_\calD$, and $M'\setminus R' = M\setminus R$. Since
$M$ is a winning strategy, for every $c\in A$ we have
\begin{align*}
s_c(M'\setminus R') - s_p(M'\setminus R') =  s_c(M \setminus R) - s_p(M \setminus R) \le 0;
\end{align*}
if $c\succ p$, this inequality is strict.
\item[(b)]
If $i \notin R'$, let $R = R'$. Then for every $c\in C\setminus \{a, p\}$ we have
\begin{align*}
s_c (M'\setminus R') = s_c (M \setminus R),
\end{align*}
and, since $w_j < w_i$, we obtain
\begin{align*}
s_a (M'\setminus R') &= s_a (M \setminus R) + w_j - w_{i} < s_a (M \setminus R),\\
s_p (M'\setminus R') &= s_p (M \setminus R) - w_{i} + w_j > s_p (M \setminus R).
\end{align*}
Combining these facts, for every $c\in A$ we have
\begin{align*}
s_c (M'\setminus R') - s_p(M'\setminus R') < s_c (M \setminus R) - s_p(M \setminus R) \le 0.
\end{align*}
Thus, both in case (a) and in case (b), 
$p$ remains the winner after recounting.
\end{itemize}
This completes the proof of the lemma.
\end{proof}
We have described an algorithm that finds a winning regular manipulation 
(and returns $\varnothing$ if no such manipulation exists) 
in polynomial time. Thus, {\sc $\PD$-Man-Reg} is in P.
\end{proof}

\section{Conclusion and Open Problems}
We have studied the problem of protecting elections by means of recounting
votes in the manipulated districts. Our results offer an almost complete picture of the
worst-case complexity of the problems faced by the defender and the attacker. Perhaps
the most obvious open question is whether we can strengthen the NP-hardness results
for {\sc $\PV$-Man} and for {\sc $\PD$-Man} under unary representation to
$\Sigma_2^P$-completeness results. The next challenge is to extend our results
beyond Plurality; e.g., leadership elections are often conducted using Plurality
with Runoff, and it would be interesting to understand if similar results
hold for this rule.

Our model is quite expressive: districts may have different weights, and an attacker
may only be able to corrupt a fraction of votes in a district. These features of the model
are intended to capture the challenges of real-world scenarios; in particular, it is typically
infeasible for the attacker to change {\em all} votes in a district.
However, it is important to understand their impact on the complexity of the
problems we consider. We tried to indicate which of our hardness results hold for special
cases of the model, and proved some easiness results under simplifying assumptions,
but it would be good to obtain a more detailed picture, possibly using the tools
of parameterized complexity. A concrete open question is whether our $\Sigma_2^P$-hardness
result holds if $\gamma_i=n_i$ for all $i\in [k]$. 

We contrasted out model with that of \citet{yin2018optimal}, where the defender
moves first and protects some of the districts from manipulation. In practice, the defender
can use a variety of protective measures at different points in time, and an exciting direction
for future work is to analyze what happens when the defender can split her resources
among different activities, with some activities preceding the attack, and others
(such as recounting) undertaken in the aftermath of the attack.

\bibliographystyle{named}
\bibliography{recounting-bib}

\newpage
\appendix

\section{Appendix: Proof of Lemma~\ref{lem:SSS-hardness}}

\begin{figure*}[t]
\newcolumntype{R}{>{$}r<{$}}
\newcolumntype{C}{>{$}c<{$}}
\newcommand{\setline}{\rule[-0.4em]{0pt}{1.5em}}
\newcommand{\addblankrow}{\multicolumn{7}{r}{\rule[-0.3em]{0pt}{1em}}\\}
\newcommand{\setcell}[1]{\multicolumn{1}{c}{\makebox[7mm][c]{#1}}}
\small
\centering
\begin{tabular}{R|R|R|R|R|R|R|l}
\multicolumn{1}{C}{} & \setcell{$1$} 	&  \multicolumn{1}{c}{\makebox[16mm][c]{$\dots$}}  &  \setcell{$2i$}  & \setcell{$2i+1$}  &  \multicolumn{1}{c}{\makebox[16mm][c]{$\dots$}}  & \setcell{$2|B|+2$}  &  \\[1mm]
\hhline{~|------|~}
\setline x_i & \cellcolor{gray!25}	\overline{1}	 &&&&	    &  \cellcolor{gray!25} \overline{a_i}  &\emph{(one copy for each $i \in \{1,\dots,|A|\}$)}  \\\hhline{~|------|~}
\addblankrow\hhline{~|------|~}
\setline x_0 &	\cellcolor{gray!25} \overline{1}	 &&&&	    & \cellcolor{gray!25} \overline{0}  &\emph{($q$ copies)}  \\\cline{2-7}
\addblankrow[3mm]
\hhline{~|------|~}
\setline y_i &&& \cellcolor{gray!25} \overline{1} &&	    &   \cellcolor{gray!25}  \overline{b_i} &\emph{($q+1$ copies for each $i \in \{1,\dots,|B|\}$)}  \\\hhline{~|------|~}
\addblankrow\hhline{~|------|~}
\setline y'_i &&& \cellcolor{gray!25} \overline{1} &&	    &    \cellcolor{gray!25} \overline{0} &\emph{($q+1$ copies for each $i \in \{1,\dots,|B|\}$)}  \\\hhline{~|------|~}
\addblankrow[3mm]\hhline{~|------|~}
\setline w_i &&&& \cellcolor{gray!25} \overline{1} &	    &    & %
    \multirow{3}{*}{%
    $\left\}\  \begin{array}{l}
     \text{\em ($q$ copies of $y_{-i} = w_i-z_i$}
    \\[0.7mm] \text{\em  for each $i \in \{1,\dots,|B|\}$)}
    \end{array}\right.$ 
    }
\\\hhline{~|------|~}
\multicolumn{6}{r}{}\\[-0.5mm]\hhline{~|------|~}
\setline z_i &&&&  &	    &   \cellcolor{gray!25}  \overline{b_i} &  \\\hhline{~|------|~}
\addblankrow\addblankrow
\noalign{\hrule height 0.8pt width 120mm}
\addblankrow\addblankrow\hhline{~|------|~}
\setline s & \cellcolor{gray!25} \overline{q} & \multicolumn{1}{c|}{\cellcolor{gray!25} $\cdots$} &
\cellcolor{gray!25} \overline{2q+1} & \cellcolor{gray!25} \overline{q} &
\multicolumn{1}{c|}{\cellcolor{gray!25} $\cdots$}	
&   \cellcolor{gray!25} \overline{t} & \emph{(the goal)}  \\\cline{2-7}
\addblankrow
\end{tabular}
\caption{Reduction from {\sc bi-SS} to {\sc SSS$^+$} (all blank sections are $\overline{0}$s).
\label{fig:reduction-biSS-to-SSS-t}}
\end{figure*}

In order to prove that SSS is $\Sigma_2^P$-complete, we will first show
(Lemma~\ref{lem:complexity-SSS-t}) that a variant of this problem,
which we call {\sc SSS$^+$}, is $\Sigma_2^P$-complete.
We will then explain how to reduce {\sc SSS$^+$} to SSS.

An instance of {\sc SSS$^+$} is given by a set of positive integers $Y$
and two positive integers $r$ and $f$.
It is a yes-instance if there exists a subset $Y' \subseteq Y$ with $|Y'|=r$
such that for all $Y'' \subseteq Y'$ it holds that $\sum_{x \in Y''} x \neq f$,
and a no-instance otherwise.

\begin{lemma}\label{lem:complexity-SSS-t}
{\sc SSS$^+$} is $\Sigma_2^P$-complete.
\end{lemma}

\begin{proof}
It is easy to see that {\sc SSS$^+$} is in $\Sigma_2^P$.
In the remainder of the proof, we show that this problem is $\Sigma_2^P$-hard.
Consider an instance $\langle Y, r, f\rangle$ of {\sc SSS$^+$}.
Let $q=|Y|-r$, $s = \sum_{x\in Y} x- f$.
Note that $\langle Y, r, f\rangle$ is a yes-instance of {\sc SSS$^+$}
if and only if there exists a subset $Y' \subseteq Y$
with $|Y'|= q$ such that for all $Y'' \subseteq Y\setminus Y'$
it holds that $\sum_{y \in Y'} y + \sum_{y \in Y''} y \neq s$;
thus, this instance of {\sc SSS$^+$} can be equivalently described
by the triple $\langle Y, q, s\rangle$.

To prove hardness of {\sc SSS$^+$},
we show a reduction from the {\sc bilevel Subset Sum} ({\sc bi-SS}) problem, which is known to be
$\Sigma_2^P$-complete~\citep{berman2002complexity}.

\begin{definition}[{\sc bilevel Subset Sum ({\sc bi-SS})}]
An instance of {\sc bi-SS} is given
by a positive integer $t$
and two sets of positive integers $A$ and $B$.
It is a yes-instance if there exists a set $A'\subseteq A$ such that for
all $B'\subseteq B$ it holds that $\sum_{a\in A'} a + \sum_{b\in B'} b \neq t$, and a no-instance otherwise.
\end{definition}

It is convenient to think of both {\sc bi-SS} and
{\sc SSS$^+$} as leader-follower games.
The leader acts first by selecting a subset; his aim is to prevent
the sum of the integers chosen by both players from reaching a given target.
The follower acts second; her aim is to select a subset so that the sum
of the chosen integers equals the target. The difference between these two games
is that in the former game the leader and the follower select from two different sets
and there is no limit of the number of integers each of them can choose,
whereas in the latter game the leader is limited to $q$ integers and both parties
choose from the same base set.

Given an instance $\langle A, B, t\rangle$ of {\sc bi-SS},
where $A=\{a_1, \dots, a_{|A|}\}$, $B=\{b_1, \dots, b_{|B|}\}$,
we proceed as follows.
We will represent a positive integer $x$ as a vector of bits of length $L=\lfloor\log_2 x\rfloor+1$,
denoted $\overline{x}=\overline{x^1\dots x^L}$: we have $\sum_{i \in [L]} \overline{x^i} \cdot 2^{L-i} = x$.
We will consider numbers that correspond to bit vectors consisting of $2|B| + 2$ sections,
with each section consisting of
$$
\left\lceil \log_2\left(\sum_{a\in A} a + (|A|+1)\sum_{a\in B} b + 1\right) \right\rceil +
\left\lceil \log_2 (2|A|+2) \right\rceil
$$
bits; this value, which is polynomial in the size of the input,
is chosen so that addition operations do not carry bits across sections.
For $h=1, \dots, 2|B|+2$, let $\overline{x}(h)$ denote the $h$-th section of $\overline{x}$.

We now construct an instance of {\sc SSS$^+$} described by a triple
$\langle Y, q, s\rangle$.
Let $q = |A|$. The set $Y$ consists of the following integers
(see also Figure~\ref{fig:reduction-biSS-to-SSS-t}):
\begin{itemize}
\item
For each $i = 1, \dots, q$,
there is an integer $x_i$ such that $\overline{x_i}(1) = \overline{1}$,
$\overline{x_i}(2|B| + 2) = \overline{a_i}$, and
$\overline{x_i}(h) = \overline{0}$ for each section $h \neq 1, 2|B|+2$.

\item
There are $q$ copies of integer $x_0$ such that $\overline{x_0}(1) = \overline{1}$,
and $\overline{x_0}(h) = \overline{0}$ for each section $h \neq 1$.

\item For each $i = 1, \dots, |B|$, there are:

\begin{itemize}
\item
$q+1$ copies of integer $y_i$ such that
$\overline{y_i}(2i) = \overline{1}$,
$\overline{y_i}(2|B|+2) = \overline{b_i}$, and
$\overline{y_i}(h) = \overline{0}$ for each section $h \neq 2i, 2|B|+2$.

\item
$q+1$ copies of integer $y'_i$ such that
$\overline{y'_i}(2i) = \overline{1}$ and
$\overline{y'_i}(h) = \overline{0}$ for every section $h\neq 2i$.

\item $q$ copies of integer $y_{-i} = w_i - z_i$,
where $w_i$ is such that
$\overline{w_i}(2i+1) = \overline{1}$ and
$\overline{w_i}(h) = \overline{0}$ for every $h \neq 2i+1$,
while $z_i$ is such that
$\overline{z_i}(2|B|+2) =\overline{b_i}$ and
$\overline{z_i}(h) = \overline{0}$ for every $h \neq 2|B|+2$.
\end{itemize}

\end{itemize}
Also, we set the goal $s$ so that
$\overline{s}(1) = \overline{q}$,
$\overline{s}(2|B|+2) = \overline{t}$, and
$\overline{s}(2h) = \overline{2q+1}$,
$\overline{s}(2h+1) = \overline{q}$ for each $h\in\{1, \dots, |B|\}$.


To verify the correctness of the reduction, we first make the following observation.
In the {\sc SSS$^+$} instance, the follower can achieve the goal $s$ only if,
for each $i = 1, \dots, |B|$, all copies of $y_{-i}$ and
exactly $2q+1$ out of the $2q+2$ copies of $y_i$ and $y'_i$
are included in the set $Y' \cup Y''$, which is chosen by the joint efforts
of the leader and the follower: otherwise, the $2i$-th and the $(2i+1)$-th
sections of the sum would not match the corresponding sections in $s$.
The follower can decide whether $Y'\cup Y''$ will contain $q+1$ copy of $y_i$
and $q$ copies of $y'_i$ or vice versa,
since the leader's choice is restricted to $q$ integers, while the follower's choice is unrestricted.
Therefore, for each $i = 1, \dots, |B|$, the $(2|B|+2)$-th section of the sum
of the selected copies of $y_i$, $y'_i$ and $y_{-i}$ will be either
$\overline{0}$ or $\overline{b_i}$; effectively, the follower chooses whether to include $b_i$ in the sum.

Now, suppose that in the given {\sc bi-SS} instance there exists a subset $A' \subseteq A$
such that for all $B' \subseteq B$ it holds that $\sum_{a\in A'} a + \sum_{b\in B'} b \neq t$.
Then, in the corresponding instance of {\sc SSS$^+$}
the leader can choose the subset $Y'$ containing all $x_i$ such that $a_i \in A'$
and $q - |A'|$ copies of $x_0$. Given this choice of the leader,
the follower can only choose integers from the copies of $y_i$, $y'_i$, and $y_{-i}$
since any other choice will make the first section of the sum different from $\overline{q}$.
However, since $\sum_{a\in A'} a + \sum_{b\in B'} b \neq t$ for all $B' \subseteq B$, no matter which
integers the follower chooses, the last section of the sum cannot be $\overline{t}$.
Thus, this instance of {\sc SSS$^+$} is a yes-instance.

Conversely, suppose that the {\sc bi-SS instance} is such that
for every $A' \subseteq A$ there exists a $B' \subseteq B$ such that
$\sum_{a\in A'} a + \sum_{b\in B'} b = t$. We will now argue that in the
corresponding instance of {\sc SSS$^+$} the follower can always achieve the goal $s$.
Indeed, suppose the leader chooses a set $Y'$. Let $A'=\{a_i: x_i\in Y'\}$,
and, for each $i\in \{1, \dots, |B|\}$, let $\alpha_i$ be the number of copies
of $y_i$ in $Y'$, let $\alpha'_i$ be the number of copies of $y'_i$ in $Y'$,
and let $\beta_i$ be the number of copies of $y_{-i}$ in $Y'$.
Fix some set $B'\subseteq B$ such that $\sum_{a\in A'} a + \sum_{b\in B'} b = t$.
To achieve the goal $s$, the follower can include the following integers in $Y''$:
\begin{itemize}
\item
$q - |A'|$ copies of $x_0$, so that the first section of the sum is exactly $\overline{q}$;
\item
$q+1-\alpha_i$ copies of $y_i$, $q-\alpha'_i$ copies of $y'_i$, and $q-\beta_i$
copies of $y_{-i}$ for each $i$ such that $b_i \in B'$,
so that the last sections of the copies of $y_i, y'_i$, and $y_{-i}$
in $Y'\cup Y''$ sum up to $\overline{b_i}$;
\item
$q-\alpha_i$ copies of $y_i$, $q+1-\alpha'_i$ copies of $y'_i$, and $q-\beta_i$ copies of $y_{-i}$ for each $i$
such that $b_i \notin B'$,
so that the last sections of the copies of $y_i, y'_i$, and $y_{-i}$
in $Y'\cup Y''$
sum up to $\overline{0}$.
\end{itemize}
This completes the proof.
\end{proof}

We are now ready to show that SSS is $\Sigma_2^P$-complete.
This problem is obviously in $\Sigma_2^P$. We show that it is $\Sigma_2^P$-hard
via a reduction from {\sc SSS$^+$}.
Given an instance $\langle Y, r, f \rangle$ of {\sc SSS$^+$},
we construct an instance $\langle X, \ell \rangle$ of SSS as follows.
Let $q = |Y| - r$, $z = \sum_{x\in Y} x + 1$ and $z' = -f - (q+1)z$.
\begin{itemize}
\item
Let $X$ consist of all integers in $Y$,
$2q+1$ copies of $z$ and $q+1$ copies of $z'$. Thus, $|X| = |Y| + 3q + 2$.
\item
Set $\ell = |Y| + 2q + 2$.
\end{itemize}
Observe that any subset $X' \subseteq X$ of size $\ell$ must contain:
(1) at least $r$ integers from $Y$,
(2) at least one copy of $z'$, and
(3) at least $q+1$ copies of $z$.
Note also that $z'+(q+1)z = -f$.
We will show that $\langle Y, r, f \rangle$
is a yes-instance of {\sc SSS$^+$} if and only if
$\langle X, \ell \rangle$ is a yes-instance of SSS.

Suppose that $\langle X, \ell \rangle$ is a yes-instance of SSS;
thus, there exists $X' \subseteq X$ such that
$|X'| = \ell$ and $\sum_{x \in X''} x \neq 0$
for all $X'' \subseteq X'$ with $X''\neq\varnothing$.
By our observation, $X'$ contains at least $r$ integers from $Y$.
Consider a set $Y'$ obtained by picking exactly $r$ elements from $X' \cap Y$.
We claim that $\sum_{x \in Y''}x \neq f$ for all $Y'' \subseteq Y'$.
Indeed, pick an arbitrary subset $Y''\subseteq Y'$, and consider
a set $X''$ containing all integers in $Y''$,
one copy of $z'$ and $q + 1$ copies of $z$.
If $\sum_{x \in Y''} x = f$, then $\sum_{x \in X''} x = 0$,
contradicting our assumption that $\langle X, \ell \rangle$ is a yes-instance of SSS.
Therefore, $Y'$ witnesses that $\langle Y, r, f \rangle$ is a yes-instance of {\sc SSS$^+$}.

Conversely, suppose that $\langle Y, r, f \rangle$ is a yes-instance of {\sc SSS$^+$}, i.e.,
there exists a set $Y' \subseteq Y$ such that $|Y'| = r$ and
$\sum_{x \in Y''}x \neq f$ for all $Y'' \subseteq Y'$.
Consider the set $X'$ containing all integers in $Y'$,
all $q+1$ copies of $z'$, and $q+1$ copies of $z$; hence,
$|X'| = r + 2q + 2 = \ell$.
Suppose towards a contradiction that $\sum_{x \in X''} x = 0$ for some
$X'' \subseteq X'$ with $X'' \neq \varnothing$.

Let $n_z$ and $n_{z'}$ be the number of copies of $z$ and $z'$
in $X''$, respectively; we have $n_z, n_{z'}\le q+1$.
We have
\begin{align}\label{eq:complexity-sss:expand-z}
& 0 = \sum_{x\in X''} x = n_{z'} \cdot z' + n_z \cdot z + \sum_{x \in X'' \cap Y'} x.
\end{align}
Since $z$ and all numbers in $Y'$ are positive,
it holds that $n_{z'}>0$.
Substituting $z'=-f-(q+1)z$, we obtain
\begin{align*}
n_{z'} f + [(q+1)n_{z'}- n_z] z
= \sum_{x \in X'' \cap Y'} x
\le  \sum_{x \in Y} x < z
\end{align*}
and hence
\begin{align}
&(q+1)n_{z'} - n_z <  1. \label{eq:complexity-sss:expand-x-2}
\end{align}
Since $n_z, n_{z'} \le q + 1$ and $n_{z'}>0$,
Eq.~\eqref{eq:complexity-sss:expand-x-2} implies that $n_{z'}=1$ and $n_z=q+1$.
Substituting these values into Eq.~\eqref{eq:complexity-sss:expand-z},
we obtain $\sum_{x \in X'' \cap Y'} x = f$, which contradicts the assumption that
$\sum_{x \in Y''} x \neq f$ for all $Y'' \subseteq Y'$.
We conclude that $\sum_{x \in X''} x \neq 0$ for all $X'' \subseteq X'$
with $X'' \neq \varnothing$, so $\langle X, \ell \rangle$ is a yes-instance of SSS.
\hfill $\qed$

\end{document}